\documentclass[12pt]{article}
\usepackage{txfonts}
\usepackage{latexsym}
\usepackage{stmaryrd}
\newtheorem{theorem}{Theorem}

\newtheorem{lemma}{Lemma}

\newtheorem{definition}{Definition}

\newenvironment{proof}{\par\noindent{\bf Proof:}}{\qed \par}

\newcommand{\blackslug}{\mbox{\hskip 1pt \vrule width 4pt height 8pt 
depth 1.5pt \hskip 1pt}}
\newcommand{\qed}{\quad\blackslug\lower 8.5pt\null\par\noindent}
\newcommand{\eqdef}{\stackrel{\rm def}{=}}

\title{A substructural logic for quantum measurements}

\author{Daniel Lehmann\thanks{School of Computer Science and Engineering, Hebrew University, 
Jerusalem 91904, Israel email: lehmann@cs.huji.ac.il}
}
\date{July 2023}

\begin{document}
\maketitle
\begin{abstract}
This paper presents a substructural logic of sequents with very restricted exchange and
weakening rules.
It is sound with respect to sequences of measurements of a quantic system.
A sound and complete semantics is provided.
The semantic structures include a binary relation that expresses orthogonality between 
elements and enables the definition of an operation that generalizes the projection operation in 
Hilbert spaces.
The language has a unitary connective, a sort of negation,
and two dual binary connectives that are neither commutative nor associative, sorts of
conjunction and disjunction.
This provides a logic for quantum measurements whose proof theory is aesthetically pleasing.
\end{abstract}

\section{Introduction} \label{sec:intro}
\subsection{Purpose of this paper} \label{sec:purpose}
This paper tries to develop a logical framework for the understanding of measurements in 
quantum mechanics. 
In quantum mechanics the information gathered by a sequence of two measurements, one
after the other, e.g., measuring a spin as $+ \frac{1}{2}$ in the $x$ direction and then 
$+ \frac{1}{2}$ in the $y$ direction, is not equivalent to the information gathered by those 
measurements performed in the opposite order: measurements do not commute.
The purpose of this paper is to study this non-commutative logic.

\subsection{Background} \label{sec:background}
Here is a description of the background on which this work has been elaborated.
Following~\cite{BirkvonNeu:36}, many proposed to elaborate Quantum Logic in 
orthomodular lattices: see, e.g., 
\cite{Finch_lattice:69, Finch_Sasaki:69,  Clark:quantum, CutlandGibbins:sequent, 
Pavicic:Biblio, Soler:95, DallaChiara:01, Coecke_Smets_hook:2001}.
This work takes the position that lattices are not rich enough structures to
model the projection operation that is central to quantum logic.
The definition and the study of structures in which the Sasaki projection generalizes 
the projection of a subspace onto another subspace is the goal of this paper.
Projection is the natural \emph{conjunction} of quantum logic, 
not the intersection considered by~\cite{BirkvonNeu:36} and 
this is the connective studied in this paper.
A good overview of the state of the art may be found in~\cite{Handbook:EGL1} and
a thorough discussion of its philosophical aspects in~\cite{deRonde+:2022}.
The present effort differs from the extant literature (see, e.g., \cite{Nishimura:JSL, 
Tokuo:2022, Kornell:arxiv_2021, Tokuo_implicational:2022, Fazio+3:2021}) 
in two fundamental ways:
\begin{itemize}
\item
the semantics of our logic is provided by structures that are much richer than orthomodular
lattices. Our O-spaces are orthomodular lattices that enjoy additional properties that bring 
them closer to Hilbert spaces, and
\item
the central connective we consider is not interpreted as intersection but as projection. 
\end{itemize}
One may remark that conjunction is not a natural connective in quantic matters: the meaning
of \emph{ the spin of the electron is $+\frac{1}{2}$ in the $x$-direction \emph{and} 
$+\frac{1}{2}$ in the $y$-direction} is unclear.
On the contrary the claim \emph{ the spin was measured $+\frac{1}{2}$ in the $x$-direction
\emph{and then} measured $+\frac{1}{2}$ in the $y$-direction} is perfectly clear.

This work owes a lot to the pioneering work of Girard on linear logic~\cite{Girard:87}.
It tries to do for Quantum Logic what Girard did for logics that satisfy the Exchange rule.
It describes Quantum Logic as a substructural logic in which the sides of a sequent are 
ordered sequences and not multi-sets as in Linear Logic,
see~\cite{Girard:2004, Girard:McGill, Yetter:LL}.
Like Linear Logic, our Quantum Logic admits pervasive symmetries.
The proof of the completeness result in Theorem~\ref{the:omega} uses the method 
presented in~\cite{Girard:87}.
The research stream presented 
in~\cite{LEG:Malg, Qsuperp:IJTP, SP:IJTP, Lehmann_andthen:JLC} is a precursor of this paper. 
This paper manages to circumscribe the logical part of the general endeavor.

\subsection{Plan of this paper} \label{sec:plan}
Section~\ref{sec:O-spaces} presents the mathematical structures on which we want 
to interpret our logic: O-spaces, and presents examples of such spaces.
Two dual operators are defined:
they are the basis of a non-commutative, non-associative \emph{set theory}.
In a sense, this is the elementary set theory of the quantum world.
Section~\ref{sec:O-properties} studies at length the properties of O-spaces and their
associated operations.
Section~\ref{sec:syntax} defines the language we want to use, 
the interpretation of the language over O-spaces, 
the notion of a sequent and its interpretation.
Section~\ref{sec:rules} proposes a sound and complete set of deduction rules 
for the logic of O-spaces.
Section~\ref{sec:future} reviews future lines of research.

\section{O-spaces} \label{sec:O-spaces}
This section proposes O-spaces as mathematical structures generalizing Hilbert spaces.
They will provide the semantics for the quantum logic to be described later.
The definition of O-spaces consists in five conditions that embody fundamental logical
principles. 
These conditions are, essentially, properties of the Sasaki projection that generalizes 
the projection operation of Hilbert spaces to O-spaces, as will be seen 
in Section~\ref{sec:Sasaki}.
It is remarkable that O-spaces satisfy orthomodularity 
but not the MacLane-Steinitz (MLS) exchange property.
An additional condition, described in~\cite{Lehmann-metalinear1:2022}, 
implies the MLS exchange property.
The present effort should be compared to~\cite{Domenech+:2008}.

\subsection{First concepts} \label{sec:O-def}
We shall consider structures of the type \mbox{$ \langle X , \bot , {\cal F} \rangle $} where 
the carrier $X$ is any non-empty set, \mbox{$\bot \subseteq X \times X $} is a binary relation
on $X$ and $\cal F$ is a family of subsets of $X$.
The set $X$ should be viewed as the set of possible states (of the world, or of a physical 
system) and a subset $A$ of $X$ represents a proposition: \emph{the world is in one of the 
states of the set $A$}.
The nature of \emph{a possible state} is left undefined: in the quantum world one may think
of a pure state.

In classical logic or classical physics two different states of the world exclude each other:
if the system is in state $x$, it cannot be measured in a state $y$ different from $x$.
In quantum physics, a system in state $x$ can, typically, be found, 
by a measurement, in a state $y$ different from $x$. 
The exclusion relation, in quantum physics, is orthogonality: if $x$ and $y$ are orthogonal,
then a system in state $x$ cannot be measured in state $y$.
The relation $\bot$ represents this orthogonality relation: for \mbox{$ x , y \in X $},
\mbox{$ x \, \bot \, y $} should be understood as \emph{the states $x$ and $y$ are 
orthogonal}.
Different states are not, typically, orthogonal to each other and
quantum logic must, therefore, be different from classical logic.
As expected, quantum logic will not be \emph{bivalent}: 
if $\alpha$ is the proposition
\emph{ the spin of the electron is $- \frac{1}{2}$ along the x-axis } its negation will be
\emph{ the spin is $\frac{1}{2}$ along the x-axis} but both propositions may fail to be true, if
the electron is in some superposition state.

Not all subsets of $X$ correspond to possible concepts that can be referred to by propositions. 
The family $\cal F$ is the family of subsets of $X$ that can correspond to concepts.
Definition~\ref{def:O-space} will describe the properties of $\bot$ and $\cal F$ that 
characterize O-spaces.
We need, first, to put down some definitions and fix some notations concerning binary relations.
Definition~\ref{def:first-concepts} does not assume anything on the relation $\bot$, but,
it corresponds to more or less familiar concepts only when the relation $\bot$ is 
a symmetric relation.
The reader should think of $\bot$ as the relation of orthogonality 
between vectors of an inner-product vector space.
\begin{definition} \label{def:first-concepts}
Assume a structure \mbox{$ \langle X , \bot \rangle $}.
For any \mbox{$ A , B \subseteq X $} we shall say that
$A$ and $B$ are orthogonal and write \mbox{$ A \, \bot \, B $} iff 
\mbox{$ a \, \bot \, b $} for any \mbox{$ a \in A $}, \mbox{$ b \in B $}.
We shall also write \mbox{$ x \, \bot \, A $} for \mbox{$ \{ x \} \, \bot \, A $}.
We define the orthogonal complement of any subset $A$ of $X$ by:
\begin{equation} \label{eq:complement}
A^\bot \eqdef \{ b \in X \mid b \, \bot \, a , \forall a \in A \} 
\end{equation}
and the closure of any such subset by \mbox{$ \overline{A} \eqdef {A^\bot}^\bot $}.
We shall say that \mbox{$ A\subseteq X$} is a \emph{flat} iff \mbox{$ A = \overline{A}$}.
The set \mbox{$ Z \subseteq X $} is defined by:
\mbox{$ Z \eqdef \{ y \in X \mid y \, \bot \, x , \forall x \in X \} $}, the set of all states that are
orthogonal to every state.
For any \mbox{$ x , y \in X $}, we shall say that $x$ and $y$ are \emph{equivalent} and write
\mbox{$ x \sim y $} iff \mbox{$ \{ x \}^\bot = \{ y \}^\bot $}, i.e., if $x$ and $y$ are 
orthogonal to exactly the same elements.
\end{definition}

The word \emph{flat} is borrowed from matroid theory.
Let us just say that the orthogonal complement of $A$, $A^\bot$ represents the negation of the 
proposition represented by $A$ and that the closure $\overline{A}$ represents the concept
spanned by $A$. Flats represent experimental propositions.
The meaning of the equivalence relation $\sim$ is that, for all that matters, equivalent elements
play exactly the same role in the structure, in particular they are contained in exactly 
the same flats.
The set $Z$ is the \emph{zero} set of the structure, the set of elements that are orthogonal
to every element, such as the vector $\vec{0}$ in Hilbert spaces.
One may notice that, in contrast with the philosophical tradition, dating back to the dawn of 
greek philosophy, the basic notion is taken here to be, not that of identity, but that of difference. 
This can be compared to the approach taken in~\cite{Post:1963}.

\subsection{The Sasaki projection and its dual} \label{sec:Sasaki}
We shall, now, define a binary operation, in such structures, that generalizes the projection 
operation of Hilbert spaces.
It is by this operation that we want to interpret the binary connectives of quantum logic.
The remark that projection can be formalized without the inner product should probably
be credited to Usa Sasaki~\cite{Sasaki:1954}.

In physics a measurement transforms a knowledge set, i.e., a set of states, into
another knowledge set.
In classical physics a measurement only adds some new information, i.e., 
the new knowledge set is a subset of the old set.
In quantum physics, this is not the case and a measurement is modeled by a projection,
in a Hilbert space, of a set of states onto a linear subspace.
This paper tries to abstract from the Hilbert space formalism and to discover the logic
of knowledge update behind quantum physics.
Therefore we must consider an operation on two flats that models projection of subspaces
in Hilbert space. The family $\cal F$ is not involved in the following definition.
The reader should think of the projection of a singleton $\{ x \}$ on a flat $A$ as the projection
of the vector $\vec{x}$ on $A$, renormalized and up to an arbitrary phase factor.
 
\begin{definition} \label{def:projection}
Assume a structure \mbox{$ \langle X , \bot \rangle $}.
Let \mbox{$A , B \subseteq X $}.
We shall define the projection of $A$ onto $B$ and its dual by
\begin{equation} \label{eq:proj}
A \otimes B \eqdef \overline{B} \cap {( \overline{B} \cap A^\bot)}^\bot  \ \ \ , \ \ \ 
A \oplus B \eqdef {( B^\bot \otimes A^\bot)}^\bot.
\end{equation}
\end{definition}

Since those operations are not commutative, we have a choice to make when defining the dual:
keep the order of the operands or use the reverse order.
For formatting reasons, which will be clear in Section~\ref{sec:soundness} we choose the
reverse order. 
We shall assume that the operator $\otimes$ has precedence over $\oplus$ and that 
both $\otimes$ and $\oplus$ have precedence over union and intersection.

\subsection{O-spaces} \label{sec:O-space-def}
Definition~\ref{def:O-space} defines structures that generalize sets with complementation and 
intersection and also inner-product vector spaces under orthogonal completion and projection.
\begin{definition} \label{def:O-space}
A structure \mbox{$ \langle X , \bot , {\cal F} \rangle $} is said to be an \emph{O-space} iff:
\begin{enumerate}
\item \label{conditionS}
{\bf S \ \ \ } the relation $\bot$ is \emph{symmetric}: 
for any \mbox{$ x , y \in X $}, \mbox{$ x \, \bot \, y $} iff \mbox{$ y \, \bot \, x $},
\item \label{conditionZ}
{\bf Z \ \ \ } \ for any \mbox{$ x \in X $} such that \mbox{$ x \, \bot \, x $}, one has
\mbox{$ x \in Z $},
\item \label{conditionF}
{\bf F \ \ \ } every member of the family $\cal F$ is a flat, 
for any \mbox{$ x \in X $} one has \mbox{$\overline{ \{ x \} } \in {\cal F} $}, 
\mbox{$ Z \in {\cal F} $},
if \mbox{$ A \in {\cal F} $} then \mbox{$ A^\bot \in {\cal F} $}, 
if \mbox{$ A , B \in {\cal F} $} then \mbox{$ A \otimes B \in {\cal F} $},
\item \label{conditionO}
{\bf O \ \ \ } for any \mbox{$x \in X $} and any \mbox{$ A \in {\cal F} $}
\[
x \in \overline{ \{ x \} \otimes A \cup \{ x \} \otimes A^\bot }
{\rm \ and \ } \{ x \} \otimes A \subseteq 
\overline{ \{ x \} \cup \{ x \} \otimes A^\bot },
\]
\item \label{conditionA}
{\bf A \ \ \ } for any \mbox{$ A , B \in {\cal F} $}, \mbox{$ A \otimes B \subseteq $}
\mbox{$ \bigcup_{a \in A} \{ a \} \otimes B $}.
\end{enumerate}
\end{definition}

O-spaces provide an abstraction of Hilbert spaces that is devoid of numbers.
This fits the author's desire to describe the logical aspects of QM.
Each of the five defining properties of O-spaces can be straightforwardly justified on ontological
or epistemological grounds.
For {\bf S}, if $x$ cannot be confused with $y$, then $y$ cannot be confused with $x$.
{\bf Z} states that an element that is orthogonal to itself, i.e., different from itself, 
is different from any element.
Those principles parallel the law of identity, so fundamental in philosophical logic.
The requirement that $\bot$ be irreflexive, i.e., that no element be distinguished 
from itself, implies {\bf Z}.
The weaker {\bf Z} fits better both the orthogonality relation in Hilbert spaces as shown 
in Section~\ref{sec:Hilbert} and the logical space described 
in Section~\ref{sec:completeness}.
Property {\bf F} enumerates our conditions on the family $\cal F$.
Note that we do not require that $A \cap B$ be an element of $\cal F$ 
when $A$ and $B$ are. 
We cannot require Property {\bf O} from all flats, since in the model of 
Section~\ref{sec:completeness} used to prove completeness the property holds only
for flats that are defined by a single proposition.
Suppose $A$ and $B$ represent quantic properties.
A system in state $x$ can be tested for property $A$: either the answer will be 
\emph{yes} and the state will be an element of \mbox{$ \{ x \} \otimes A $} or 
the answer will be \emph{no} and the state will be an element of 
\mbox{$ \{ x \} \otimes A^\bot $}.
{\bf O} expresses the fact that state $x$ is a superposition of the two projections.
Similarly, $x \otimes A$ is a superposition of $x$ and $x \otimes A^{\bot}$.
{\bf A} expresses the idea that the projection of a set of states is the set 
of the projections of the individual states.

\subsection{Examples} \label{sec:examples}
Before proceeding with the study of O-spaces, 
we shall describe some examples of O-spaces.

\subsubsection{Inner-product spaces} \label{sec:Hilbert}
Hilbert spaces constitute our main inspiration for Definition~\ref{def:O-space}.
Since norms and completeness do not play a role here, all we are going to say holds 
for inner-product spaces in general, but we shall use the Hilbert space terminology 
since that is what is used in QM.
Let \mbox{$ X = \cal H$} be a Hilbert space and set \mbox{$ \vec{a} \, \bot \, \vec{b} $} iff 
\mbox{$ \langle \vec{a} \mid \vec{b} \rangle = 0 $} for any $\vec{a}$ and $\vec{b}$ 
in $\cal H$ and let $\cal F$ be the set of all subspaces of $X$.
For any \mbox{$ A \subseteq X $}, $A^\bot$ is the subspace including all vectors
orthogonal to all elements of $A$ and $\overline{A}$ is the linear subspace spanned
by the vectors of $A$, i.e., the minimal linear subspace including all elements of $A$.
Flats correspond exactly to linear subspaces.
{\bf S} holds true.
{\bf Z} holds because the only vector orthogonal to itself is the zero vector and the zero
vector is orthogonal to any vector.
Equations~(\ref{eq:proj}) imply that \mbox{$ A \otimes B $} is the subspace that is the
projection of the subspace spanned by $A$ onto the subspace spanned by $B$ and
therefore {\bf F} is satisfied.
{\bf O} is then seen to hold: any vector is a linear combination of its projections 
on any subspace $A$ and on its orthogonal complement $A^\bot$.
Also the projection of a vector $\vec{v}$ on $A$ is a linear combination of $\vec{v}$ 
and the projection of $\vec{v}$ on $A^\bot$.
{\bf A} is easily seen to hold.

\subsubsection{Sets} \label{sec:sets}
If Hilbert spaces form the framework in which quantum physics is described, 
elementary set theory is the framework for classical physics: the states of a system are sets
of properties.
Even though it is clear that classical systems are quantum systems with certain specific 
properties it is almost impossible to describe classical systems in the quantum framework and 
one does not know how to describe quantum physics without using, as a guide, 
the classical framework.
It must be, therefore, comforting that the logic of classical systems (elementary set theory) is
a special case of O-spaces.

Let $X$ be any non-empty set. 
Consider the structure \mbox{$ \langle X , \neq , {2}^{X} \rangle $}.
For any \mbox{$ A \subseteq X $}, \mbox{$ A^\bot =$} \mbox{$ X - A $}.
For two subsets \mbox{$ A , B \subseteq X $}, \mbox{$ A \, \bot \, B $} iff 
\mbox{$ A \cap B = \emptyset $} and \mbox{$ \overline{A} = A $}.
$\cal F$ is taken to be ${2}^{X}$.
Any subset of $X$ is a flat.
For any sets $A$ and $B$, \mbox{$ A \otimes B = A \cap B $} and 
\mbox{$ A \oplus B = A \cup B $}.
One sees easily that the conditions of Definitions~\ref{def:O-space} are satisfied.

\subsubsection{Subsets} \label{sec:subsets}
This example was suggested by Menachem Magidor.
Let $Y$ be any set. Let \mbox{$ X = {2}^{Y} $} and let, for any \mbox{$ x , y \subseteq Y $}, 
\mbox{$ x \, \bot \, y $} iff \mbox{$ x \cap y = \emptyset $}.
Let $\cal F$ be the set of all flats.
The relation $\bot$ is symmetric and, for any \mbox{$ x \in X $}, 
\mbox{$ x \cap x = \emptyset $} implies \mbox{$ x = \emptyset $} and 
\mbox{$ \emptyset \, \bot \, y $} for any \mbox{$ y \in X $}, 
which proves that {\bf Z} is satisfied.
For any \mbox{$ A \subseteq X $}, let \mbox{$ \tilde{A} = \bigcup_{a \in A} \{ a \} $} and
let \mbox{$\widehat{A} \in X $} be 
\mbox{$ \widehat{A} \eqdef $} \mbox{$ \{ B \subseteq Y \mid B \subseteq \tilde{A} \} $}.
For any \mbox{$ A \subseteq X $}, we have: 
\mbox{$ A^\bot =$} \mbox{$ \widehat{Y - \tilde{A}} $}, the set of all subsets of $Y$ 
that do not intersect any of the elements of $A$, the set of all subsets 
of the complement of $\tilde{A}$.
Then, $\overline{A}$ is $\widehat{\tilde{A}}$.
A set \mbox{$ A \subseteq X $} is a flat iff it is of the form $2^{B}$ for some 
\mbox{$B \subseteq Y $}.
For any \mbox{$ A , B $}, 
\[
\widehat{\tilde{B}} \cap \widehat{ Y - \tilde{A}} = \widehat{ \tilde{B} - \tilde{A} } 
{\rm \ and \ } A \otimes B = 
\widehat{\tilde{B}} \cap \widehat{ ( Y - \widehat{ \tilde{B} - \tilde{A} } ) } = 
\widehat{ \tilde{A} \cap \tilde{B} }.
\]
If $A$ and $B$ are flats, then \mbox{$A \otimes B = $} \mbox{$ A \cap B = $}
\mbox{$ \widehat{ \tilde{A} \cap \tilde{B} } $}.
Property {\bf O} is easily seen to hold.
Let $A$ and $B$ be flats, \mbox{$ A = 2^{A'} $} and \mbox{$ B = 2^{B'} $}.
\mbox{$ A \otimes B = 2^{A' \cap B'} $}.
For any \mbox{$ a \in A $}, \mbox{$ a \subseteq A' $}, \mbox{$ \overline{ \{ a \} } = $}
\mbox{$ 2^a$}  and \mbox{$ \{ a \} \otimes B = $} \mbox{$ \widehat{ a \cap \tilde{B} } $}.
If \mbox{$ x \in A \otimes B $}, \mbox{$ x \subseteq A' \cap B' $} and 
\mbox{$ x \subseteq x \cap B' = $}, \mbox{$ \tilde{ \{ x \} } \cap \tilde{B} $} and
\mbox{$ x \in \{ x \} \otimes B $}.
We have shown that \mbox{$ \bigcup_{a \in A} \{ a \} \otimes B = $} 
\mbox{$ \overline{B} \cap {( \overline{B} \cap A^\bot )}^\bot $}, i.e., that property 
{\bf A} holds true.

\subsubsection{Union} \label{sec:union}
Given any two O-spaces: \mbox{$ \langle X_{0} , \bot_{0} , {\cal F_{0}} \rangle $} and 
\mbox{$ \langle X_{1} , \bot_{1} , {\cal F_{1}} \rangle $} where 
\mbox{$ X_{0} \cap X_{1} = \emptyset $} one can define their sum by 
\mbox{$ X = X_{0} \cup X_{1} $} and
\begin{itemize}
\item 
\mbox{$ x \, \bot \, y $} iff \mbox{$ x \, \bot_{i} \, y $} for \mbox{$ x , y \in X_{i} $},
\mbox{$ i = 0 , 1 $}, and
\item
\mbox{$ x \, \bot \, y $} for \mbox{$ x \in X_{i} $} and \mbox{$ y \in X_{i + 1} $} where
addition is understood modulo $2$.
\item
\mbox{$ {\cal F} = \{ A \cup B \mid A \in {\cal F_{0}} , B \in {\cal F_{1}} \} $}.
\end{itemize}
For any \mbox{$ A \subseteq X $}, 
\mbox{$ A^\bot =$}
\mbox{$ {( A \cap X_{0} )}^{\bot_{0}} \cup {( A \cap X_{1} )}^{\bot_{1}} $} and therefore 
we also have \mbox{$ \overline{A} = $} 
\mbox{$ { {( A \cap X_{0} )}^{\bot_{0}}}^{\bot_{0}} \cup 
{ {( A \cap X_{1} )}^{\bot_{1}}}^{\bot_{1}} $}.

One easily sees that the sum of two O-spaces is an O-space.
In fact, one may define infinite sums in the same way and with the same property.

The union construction enables the definition of structures that are fit 
for the description of quantic systems with a superselection rule that forbids 
the existence of certain superpositions.
A particle whose total spin is unknown cannot be found in a superposition of a particle
of total spin $\frac{1}{2}$ and a particle of spin $1$.
The union of a two-dimensional and a three dimensional O-space is suitable to represent
such a system: there are no superpositions of states in different parts of the union.

\section{Properties of O-spaces} \label{sec:O-properties}
We shall present a wealth of results concerning O-spaces.
In particular we shall show that O-spaces satisfy orthomodularity. 
We shall introduce the five properties of Definition~\ref{def:O-space} in two batches.
First, we consider a structure \mbox{$ \langle X , \bot \rangle $} 
where $\bot$ satisfies {\bf S} and {\bf Z}.
\subsection{Consequences of {\bf S} and {\bf Z}} \label{sec:SZ}
Our main result is that, if {\bf S} is satisfied, the operation \mbox{$ A \mapsto \overline{A}$} 
is a closure operation,
but note that this property stems from the properties of its \emph{square root}, 
the orthogonal complement operation, so that it is not just any closure operation.
\begin{lemma} \label{the:symmetric}
Let \mbox{$ \langle X , \bot \rangle $} be a structure that satisfies {\bf S} and {\bf Z}.
For any \mbox{$ A , B , C \subseteq X $}, for any family, finite or infinite, 
\mbox{$A_{i} \subseteq X $}, \mbox{$ i \in I $}, any \mbox{$ x , y \in X $}:
\begin{enumerate}
\item \label{union-bot}
\mbox{$ {( \bigcup_{i \in I} A_{i} )}^\bot = $} \mbox{$ \bigcap_{i \in I} A_{i}^\bot $} and 
\mbox{$ ({ \bigcap_{i \in I} \overline{ A_{i} } )}^\bot = $} 
\mbox{$ \overline{ \bigcup_{i \in I} A_{i}^\bot } $},
\item \label{anti}
if \mbox{$ A \subseteq B $}, then 
\mbox{$ B^\bot \subseteq A^\bot $} and \mbox{$ \overline{A} \subseteq \overline{B} $},
\item \label{closure}
\mbox{$ \overline{A}^\bot =$} \mbox{$ A^\bot $} and the operation 
\mbox{$ A \mapsto \overline{A}$} is a closure operation, i.e., it is monotone,
\mbox{$ A \subseteq \overline{ A } $} and \mbox{$ \overline{ \overline{ A } } = $}
\mbox{$ \overline{ A } $},
\item \label{BxinA}
for any \mbox{$ A \subseteq X $} and any \mbox{$ x \in X $}, 
if \mbox{$ x \in A $}, then \mbox{$ x \in \{ x \} \otimes A $},
\item \label{flat}
A set \mbox{$ A \subseteq X $} is a flat iff there is a set \mbox{$ B \subseteq X $} such that
\mbox{$ A = B^\bot $}, the intersection of a family of flats is a flat and, for any sets 
\emph{not necessarily flats} $A$ and $B$,
\mbox{$ A^\bot $}, \mbox{$ A \otimes B $} and \mbox{$ A \oplus B $} are flats,
\item \label{bar-cup}
\mbox{$ \overline{ \bigcup_{i \in I} A_{i} } = $} 
\mbox{$ \overline{ \bigcup_{i \in I} \overline{ A_{i} } } $},
\item \label{Xempty} 
\mbox{$ Z^\bot = X $} and \mbox{$ X^\bot = Z $},
therefore $Z$ and $X$ are flats,
\item \label{empty-full}
\mbox{$ \overline{ A } \cap A^\bot = Z $} (implying \mbox{$ Z \subseteq \overline{A} $}),
\mbox{$ \overline{ A \cup A^\bot } = X $},
\item \label{OAUB} 
\mbox{$ A \otimes ( \overline{B} \cap \overline{ A \cup B^\bot } ) = $}
\mbox{$ A \otimes B $},
\item \label{bar2}
\mbox{$ A \otimes B =$} \mbox{$ A \otimes \overline{B} = $}
\mbox{$ \overline{A} \otimes B = $} \mbox{$ \overline{ A \otimes B } $} and 
\mbox{$ A \oplus B = $} \mbox{$ A \oplus \overline{B} = $}
\mbox{$ \overline{A} \oplus B = $} \mbox{$ \overline{ A \oplus B } $},
\item \label{overline}
\mbox{$ A \otimes B = \overline{B} \cap \overline{ B^\bot \cup A } = $}
\mbox{$ {( B^\bot \oplus A^\bot )}^\bot $} and
\mbox{$ A \oplus B = \overline{ A \cup ( A^\bot \cap \overline{B} ) } $},
\item \label{bot-union}
if \mbox{$ A \bot B $}, one has \mbox{$ A \oplus B = \overline{A \cup B} $},
\item \label{AinterB}
\mbox{$ A\cap B \subseteq A \otimes B $} and \mbox{$ A \oplus B \subseteq$}
\mbox{$ \overline{A \cup B} $},
\item \label{AsubB1}
if \mbox{$ A \subseteq B $}, then \mbox{$ B \otimes A =$} \mbox{$ \overline{A} $} and
\mbox{$ A \oplus B = $} \mbox{$ \overline{B} $}, 
\item \label{BbotA}
if \mbox{$ A \, \bot \, B $}, then \mbox{$ ( B \cup C ) \otimes A = $}
\mbox{$ C \otimes A $},
\item \label{plus-union}
if \mbox{$ A \, \bot \, B $}, then \mbox{$ A \oplus B =$} \mbox{$ \overline{ A \cup B } $} and
\mbox{$ A^\bot \otimes B^\bot = $} \mbox{$ A^\bot \cap B^\bot $}.
\end{enumerate}
\end{lemma}
\begin{proof}
\begin{enumerate}
\item 
The first equality is obvious from Definition~\ref{def:first-concepts}.
Then, \mbox{$ {( \bigcup_{i \in I} A_{i}^\bot )}^\bot = $} 
\mbox{$ \bigcap_{i \in I} \overline{ A_{i} } $} and 
\mbox{$ \overline{ {( \bigcup_{i \in I} A_{i}^\bot )} } = $} 
\mbox{$ {( \bigcap_{i \in I} \overline{ A_{i} } )}^\bot $}.
\item 
By Definition~\ref{def:first-concepts}.
\item 
If \mbox{$ a \in A $} and \mbox{$ b \in A^\bot $}, we have
\mbox{$ b \, \bot \, a $} and, by {\bf S}, \mbox{$ a \, \bot \, b $} and therefore
\mbox{$ a \in {A^\bot}^\bot = $} $\overline{A} $.
We have shown that \mbox{$ A \subseteq \overline{A} $}.
By item~\ref{anti} above we conclude that \mbox{$ {\overline{A}}^\bot \subseteq $} 
\mbox{$ A^\bot $}.
But we also see that \mbox{$ A^\bot \subseteq $} \mbox{$ {{A^\bot}^\bot}^\bot = $}
\mbox{$ \overline{A}^\bot $}. 
We have shown that \mbox{$ {{A^\bot}^\bot}^\bot = $} \mbox{$ A^\bot $}.
Therefore \mbox{$ \overline{\overline{A}} = \overline{A} $} and, by item~\ref{anti},
we have shown that the operation \mbox{$ A \mapsto \overline{A}$} is a closure operation.
\item
Obviously \mbox{$ x \, \bot \, \overline{A} \cap \{ x \}^\bot $} and therefore \mbox{$ x \in A $}
implies \mbox{$ x \in \{x \} \otimes A $}, by Equations~(\ref{eq:proj}) and 
item~\ref{closure} above.
\item 
The \emph{only if} part is obvious from the definition of a flat.
For the \emph{if} part, assume \mbox{$ A = B^\bot $}.
By item~\ref{closure} above, 
\mbox{$ \overline{A} =$} \mbox{$ B^\bot = A $}.
For any $i$, \mbox{$ \bigcap_{i} A_{i} \subseteq A_{i} $} and,
by item~\ref{anti} above, 
\mbox{$ \overline{\bigcap_{i \in I } A_{i}} \subseteq$}
\mbox{$ \bigcap_{i \in I} \overline{A_{i}} = $} \mbox{$ \bigcap_{i \in I} A_{i} $}.
Conclude by item~\ref{closure} above.
The last claim follows from Equations~(\ref{eq:proj}).
\item 
By item~\ref{closure} above, we have
\mbox{$ \overline{ \bigcup_{i \in I} A_{i}  } \subseteq $} 
\mbox{$ \overline{ \bigcup_{i \in I} \overline{ A_{i}} } $}.
But we also have 
\mbox{$ \overline{ A_{i}} \subseteq $} \mbox{$\overline{ \bigcup_{i \in I} A_{i} } $}.
Therefore
\mbox{$ \bigcup_{i \in I} \overline{A_{i}} \subseteq $} 
\mbox{$ \overline{ \bigcup_{i \in I} A_{i} } $} and
by item~\ref{closure} above we have
\mbox{$ \overline{ \bigcup_{i \in I} \overline{A_{i}} } \subseteq$} 
\mbox{$ \overline{ \bigcup_{i \in I} A_{i} } $}.
\item 
By Definition~\ref{def:first-concepts} and {\bf Z}.
\item 
If \mbox{$ x \in A^\bot \cap \overline{A} $}, \mbox{$ x \, \bot \, x $}, which implies 
\mbox{$ x \in Z $}, by {\bf Z}. 
Then, by items~\ref{union-bot}, \ref{closure} and~\ref{Xempty}, 
\mbox{$ \overline{ A \cup A^\bot } = $} 
\mbox{$ {( A^\bot \cap \overline{A} )}^\bot = $} \mbox{$ Z^\bot = X $}.
\item 
By Equations~(\ref{eq:proj}) and items~\ref{closure}, \ref{flat} and~\ref{empty-full} above:
\[
A \otimes ( \overline{B} \cap \overline{ A \cup B^\bot } ) =
\overline{B} \cap \overline{ A \cup B^\bot } \cap 
{( \overline{B} \cap \overline{ A \cup B^\bot } \cap A^\bot )}^\bot = 
\]
\[
\overline{B} \cap \overline{ A \cup B^\bot } \cap 
{( \overline{B} \cap {( A^\bot \cap \overline{B} )}^\bot \cap A^\bot )}^\bot =
\overline{B} \cap \overline{ A \cup B^\bot } \cap Z^\bot =
\]
\[
\overline{B} \cap \overline{ A \cup B^\bot } = 
\overline{B} \cap {( A^\bot \cap \overline{B} )}^\bot = A \otimes B.
\]
\item 
By Equations~(\ref{eq:proj}) and items~\ref{union-bot}, \ref{closure} and \ref{flat} above.
\item 
By Equations~(\ref{eq:proj}) and items~\ref{union-bot} and~\ref{bar2}. 
\item
One has \mbox{$ B \subseteq A^\bot $} and, by items~\ref{anti} and~\ref{closure}, 
\mbox{$ \overline{B} \subseteq A^\bot $} and, by item~\ref{overline} just above,
\mbox{$ A \oplus B = $} \mbox{$ \overline{ A \cup \overline{B} } = $}
\mbox{$ \overline{A \cup B} $} by item~\ref{bar-cup}.
\item 
By Equations~(\ref{eq:proj}) and items~\ref{closure} and~\ref{overline} above.
\item 
Assume \mbox{$ A \subseteq B $}.
By item~\ref{empty-full}, \mbox{$ \overline{A} \cap B^\bot = Z $} and therefore,
by Equations~(\ref{eq:proj}),
\mbox{$ B \otimes A = $} \mbox{$ \overline{A} \cap Z^\bot = $}
\mbox{$ \overline{A} $} by item~\ref{empty-full}.
Also, by item~\ref{anti}, \mbox{$ B^\bot \subseteq A^\bot $} and, 
by what we have just shown, \mbox{$ A^\bot \otimes B^\bot = B^\bot $}.
We see that \mbox{$ A \oplus B = \overline{B} $}.
\item 
Since \mbox{$ \overline{A} \subseteq B^\bot $}, by Equations~(\ref{eq:proj}) 
and item~\ref{union-bot}:
\[
( B \cup C ) \otimes A = \overline{A} \cap {( \overline{A} \cap B^\bot \cap C^\bot )}^\bot = 
\overline{A} \cap {( \overline{A} \cap C^\bot )}^\bot = C \otimes A.
\]
\item 
By item~\ref{overline} above, 
\mbox{$ A \oplus B = $} \mbox{$ \overline{ A \cup ( A^\bot \cap \overline{ B } ) } $}.
But, if \mbox{$ A \, \bot \, B $} we have \mbox{$ \overline{ B } \subseteq A^\bot $}
and item~\ref{bar-cup} above implies our claim.
Then, by item~\ref{overline} \mbox{$ A^\bot \otimes B^\bot = $} 
\mbox{$ {( \overline{B} \oplus \overline{A} )}^\bot = $}
\mbox{$ {( \overline{ B \cup A } )}^\bot = $}
\mbox{$ {( A \cup B )}^\bot = $}
\mbox{$ A^\bot \cap B^\bot $}.
\end{enumerate}
\end{proof}

\subsection{More properties of O-spaces} \label{sec:more}
We shall now put to work our last assumptions.
We argued in Section~\ref{sec:O-space-def} that all requirements 
of Definition~\ref{def:O-space} were of a quasi-logical nature and justifiable on 
phenomenological or epistemological ground. 
Item~\ref{orthomodularity} below states that any O-space is orthomodular, 
a property that is central to all efforts in quantum logic that followed~\cite{BirkvonNeu:36}.
Our result provides a justification or explanation for orthomodularity on the basis of four 
ontological or epistemic assumptions.
This is in line with justifications proposed 
in~\cite{Piron:1964, Jauch:foundations, Beltrametti_Cassinelli:1981, Grinbaum_ortho:2005}.
Further results deal  with a central theme in QM: commutation.
The operations $\otimes$ and $\oplus$ do not commute in general, but we point out cases 
in which they do commute.
From now on, we assume that \mbox{$ \langle X , \bot , {\cal F} \rangle $} is an O-space.
Item~\ref{left-additivity}, describes  the dependence 
of the operations $\otimes$ and $\oplus$ on their first argument: 
$\otimes$ \emph{distributes} over union and $\oplus$ \emph{distributes} over intersection.
This will be used,for example, in the soundness proof for the rule of {\bf Left-Weakening} 
in Section~\ref{sec:soundness}.
\begin{lemma} \label{the:1-dim}
For any \mbox{$ x , y \in X $} and any \mbox{$ A , A_{i} , B \in {\cal F} $}:
\begin{enumerate}
\item \label{left-additivity}
\mbox{$ ( \bigcup_{i \in I} A_{i} ) \otimes B = $}
\mbox{$ \overline{ \bigcup_{i \in I} ( A_{i} \otimes B )} $} and
\mbox{$ B \oplus \bigcap_{i \in I} A_{i} = $}
\mbox{$ \bigcap_{i \in I} ( B \oplus A_{i} ) $},
\item \label{xA}
if \mbox{$ x \in A $}, then \mbox{$ \{ x \} \otimes A = \overline{ \{ x \} } $},
\item \label{small-big}
if \mbox{$ A \subseteq B $}, \mbox{$ A \otimes B = A $},
\item \label{pS}
if \mbox{$ A \bot B $}, then \mbox{$ \overline{ A \cup B } \cap B^\bot = A $},
\item \label{orthomodularity}
{\bf Orthomodularity} \ If \mbox{$ A \subseteq B $}, 
then we have \mbox{$ B = $} 
\mbox{$ \overline{ A \cup ( A^\bot \cap B ) } $},
\item \label{commute1}
if \mbox{$ A \subseteq B $}, then \mbox{$ A \otimes B = B \otimes A $} and
\mbox{$ B \oplus A = A \oplus B $},
\item \label{neg-o-axioms1}
\mbox{$ B \cap {( A \otimes B )}^\bot = $}
\mbox{$ B \cap A^\bot $}, equivalently 
\mbox{$ B \cap ( B^\bot \oplus A^\bot ) = $}
\mbox{$ B \cap A^\bot $},
\item \label{neg-o-axioms2}
\mbox{$ B \cap {( B \otimes A )}^\bot = $}
\mbox{$ B \cap A^\bot $}, equivalently,
\mbox{$ B \cap ( A^\bot \oplus B^\bot ) = $} \mbox{$ B \cap A^\bot $},
\item \label{dna4}
\mbox{$ A \otimes ( A^\bot \oplus B ) = $} \mbox{$ A \cap B $},
\item \label{dna5}
\mbox{$ ( A^\bot \oplus B ) \otimes A = $} \mbox{$ A \cap B $},
\item \label{perp}
if \mbox{$ y \in A $}, \mbox{$ y \, \bot \, x $} iff \mbox{$ y \, \bot \, \{ x \} \otimes A $}.
\end{enumerate}
\end{lemma}
\begin{proof}
\begin{enumerate}
\item 
For any \mbox{$ a \in A $}, by Lemma~\ref{the:symmetric}, item~\ref{anti} and 
Equations~(\ref{eq:proj}), \mbox{$ \{ a \} \otimes B \subseteq A \otimes B $}.
By condition {\bf A} in Definition~\ref{def:O-space}, then
\mbox{$ \bigcup_{a \in A} ( \{ a \} \otimes B ) = $} \mbox{$ A \otimes B $}.
Our first claim now follows easily.
The last claim follows by duality, using item~\ref{bar2} in Lemma~\ref{the:symmetric}.
\item
if \mbox{$ x \in A $}, then \mbox{$ A^\bot \subseteq \{ x \}^\bot $}
and, by Equations~(\ref{eq:proj}), \mbox{$ \{ x \} \otimes A^\bot = $}
\mbox{$ A^\bot \cap {( A^\bot \cap \{ x\}^\bot )}^\bot = $}
\mbox{$ A^\bot \cap A = Z $}
by Lemma~\ref{the:symmetric}, item~\ref{empty-full}.
By {\bf O}, \mbox{$ \{ x \} \otimes A \subseteq $} 
\mbox{$ \overline { \overline{ \{ x \} } \cup \{ x \} \otimes A^\bot } = $}
\mbox{$ \overline { \overline{ \{ x \} } \cup Z } = $}
\mbox{$ \overline{ \{ x \} } $}.
By Equations~\ref{eq:proj}, \mbox{$ x \in \{ x \} \otimes A $}.
Therefore \mbox{$ \overline{ \{ x \} } \subseteq $} \mbox{$ \overline{ \{ x \} \otimes A } = $}
\mbox{$ \{ x \} \otimes A $}.
\item 
By item~\ref{xA} above, Lemma~\ref{the:symmetric}, item~\ref{bar-cup} and
Lemma~\ref{the:1-dim}, item~\ref{left-additivity}.
\item 
Assume \mbox{$ A \bot B $}. We have \mbox{$ A \subseteq B^\bot $} and, 
by item~\ref{small-big} above, \mbox{$ A \otimes B^\bot = A $}.
But, by Equations~(\ref{eq:proj}) and Lemma~\ref{the:symmetric}, items~\ref{union-bot}
and~\ref{bar-cup} we have
\[
A \otimes B^\bot = B^\bot \cap {( B^\bot \cap A^\bot )}^\bot = 
B^\bot \cap \overline{ B \cup A }.
\]
\item 
By item~\ref{pS} above, since \mbox{$ B^\bot \bot \, A $} we have
\mbox{$ \overline{ B^\bot \cup A } \cap A^\bot = B^\bot $},
\mbox{$ \overline{ {( B^\bot \cup A )}^\bot \cup A } = $} $B$ and
\mbox{$ \overline{ ( B \cap A^\bot ) \cup A } = B $}.
\item 
By Lemma~\ref{the:symmetric}, item~\ref{AsubB1} and item~\ref{small-big} above.
\item 
By Equations~(\ref{eq:proj}) and Lemma~\ref{the:symmetric}, item~\ref{union-bot},
\[
B \cap {( A \otimes B )}^\bot =
B \cap {( B \cap {( B \cap A^\bot )}^\bot )}^\bot =
B \cap \overline{ B^\bot \cup ( B \cap A^\bot ) } 
\]
One concludes by item~\ref{pS} above.
\item 
Similarly,
\[
B \cap {( B \otimes A )}^\bot =
B \cap {( A \cap {( A \cap B^\bot )}^\bot )}^\bot =
B \cap \overline{ A^\bot \cup ( A \cap B^\bot ) } 
\]
One concludes by item~\ref{pS} above.
\item 
Since, by Lemma~\ref{the:symmetric}, item~\ref{overline}, \mbox{$ A \subseteq A \oplus B $},
we have \mbox{$ ( A^\bot \oplus B ) \cap A^\bot = $} \mbox{$ A^\bot $}
and therefore, by Equations~(\ref{eq:proj}):
\[
A \otimes ( A^\bot \oplus B ) = 
( A^\bot \oplus B ) \cap {( ( A^\bot \oplus B ) \cap A^\bot )}^\bot =
(A^\bot \oplus B ) \cap A = 
A \cap {( B^\bot \otimes A )}^\bot 
\]
By item~\ref{neg-o-axioms1} above, 
\mbox{$ A \cap {( B^\bot \otimes A )}^\bot = $} 
\mbox{$ A \cap B $}.
\item
Similarly:
\[
( A^\bot \oplus B ) \otimes A = 
A \cap {(A \cap {( A^\bot \oplus B )}^\bot )}^\bot = 
A \cap {( A \cap ( B^\bot \otimes A ) )}^\bot = 
\]
\[
A \cap {( B^\bot \otimes A )}^\bot = A \cap ( A^\bot \oplus B ) = 
A \cap \overline{ A^\bot \cup ( A \cap B ) } = 
A \cap ( A \cap B ) = A \cap B.
\]
\item 
Let \mbox{$ y \in A $} and \mbox{$ x \bot y $}. 
By Equations~(\ref{eq:proj}), \mbox{$ \{ x \} \otimes A \subseteq $}
\mbox{$ {( A \cap \{ x \}^\bot )}^\bot $}, but 
\mbox{$ y \in A \cap \{ x \}^\bot $} and therefore 
\mbox{$ y \bot \{ x \} \otimes A $}.
Let now \mbox{$ y \in A $} and \mbox{$ \{ x \} \otimes A \, \bot \, y $}. 
We have, by Equations~(\ref{eq:proj}) and Lemma~\ref{the:symmetric}, item~\ref{union-bot}, 
\mbox{$ y \in A \cap \overline{ A^\bot \cup ( A \cap \{ x \}^\bot ) } $}.
By item~\ref{pS} above, we conclude that \mbox{$ y \bot x $}.
\end{enumerate}
\end{proof}

Our next result will be used to prove the soundness of the {\bf Exchange}, 
the {\bf $\neg$Atomic} and the {\bf $\neg\vee$1} rules in Section~\ref{sec:soundness}.
\begin{lemma} \label{the:commute}
For any \mbox{$A , B \in {\cal F} $}, \mbox{$ A \, \bot \, B $} iff
\mbox{$ A \otimes B = Z $}.
Therefore, if  \mbox{$ A \, \bot \, B $}, one has 
\mbox{$ A \otimes B = $} \mbox{$ B \otimes A $} and 
\mbox{$ A \oplus B =$} \mbox{$ B \oplus A $}.
\end{lemma}
\begin{proof}
For the \emph{only if} part, assume that \mbox{$ A \, \bot \, B $}, i.e., 
\mbox{$ B \subseteq A^\bot $}. 
By Lemma~\ref{the:symmetric}, item~\ref{empty-full}, one has 
\mbox{$ B \cap {( B \cap A^\bot )}^\bot = $}
\mbox{$ B \cap B^\bot = Z $}.
For the \emph{if} part, note that, by Lemma~\ref{the:symmetric}, item~\ref{union-bot}:
\mbox{$ A \otimes B = B \cap {( B \cap A^\bot )}^\bot = $}
\mbox{$ {( B^\bot \cup ( B \cap A^\bot ) )}^\bot $}.
By Lemma~\ref{the:symmetric}, item~\ref{empty-full}, then, 
\mbox{$ A \otimes B = Z $} implies 
\mbox{$ X = \overline{ B^\bot \cup ( B \cap A^\bot ) } $}.
On the one hand, then, 
\mbox{$ B \cap \overline{ B^\bot \cup ( B \cap A^\bot ) } = $} $B$.
But, on the other hand, since \mbox{$ B^\bot \, \bot \, B \cap A^\bot $},
Lemma~\ref{the:1-dim}, item~\ref{pS} implies that we have 
\mbox{$ B \cap \overline{ B^\bot \cup ( B \cap A^\bot ) } = $}
\mbox{$ B \cap A^\bot $}.
We conclude that \mbox{$ B \subseteq A^\bot $}, i.e., \mbox{$ A \, \bot \, B $}.
\end{proof}
Next is an important technical result: any element $C$ of $\cal F$ 
defines an O-subspace and the closure operation on the subspace is identical 
to the original closure operation.
\begin{lemma} \label{the:O-subspace}
Let \mbox{$ C \in {\cal F}$}.
The structure \mbox{$ \langle C , \bot_{C} , {\cal F_{C}} \rangle $}, where the relation 
$\bot_{C}$ is the restriction of $\bot$ to the set $C$ and \mbox{$ {\cal F_{C}} = $}
\mbox{$ {\cal F} \cap {2}^{C} $} is an O-space.
For any \mbox{$ A , B \subseteq C $}, \mbox{$ A^{\bot_{C}} = A^\bot \cap C $},
\mbox{$ {A^{\bot_{C}}}^{\bot_{C}} = $} 
\mbox{$ \overline{A} $}, \mbox{$ A \otimes_{C} B = $}
\mbox{$ A \otimes B $} and \mbox{$ A \oplus_{C} B =$} \mbox{$ A \oplus B $}.
\end{lemma}
\begin{proof}
Obviously, the relation $\bot_{C}$ satisfies {\bf S} and {\bf Z} and therefore satisfies 
Lemma~\ref{the:symmetric}.
The set $Z_{C}$ is equal to \mbox{$ Z \cap C $}.
For any \mbox{$ A , B \subseteq C $}:
\begin{enumerate}
\item
\mbox{$ A^{\bot_{C}} = $} \mbox{$ \{ x \in C \mid x \, \bot \, A \} = $}
\mbox{$ A^\bot \cap C $}.
\item  \label{same}
\mbox{$ {A^{\bot_{C}}}^{\bot_{C}} = $} 
\mbox{$ {( A^\bot \cap C )}^\bot \cap C = $}
\mbox{$ \overline{ A \cup C^\bot } \cap C $}.
Since \mbox{$ A \, \bot \, C^\bot $}, by Lemma~\ref{the:1-dim}, item~\ref{pS},  we have
\mbox{$ \overline{ A \cup C^\bot } \cap C = $} \mbox{$ \overline{A} $}.
\item \label{same-otimes}
By the above, \mbox{$ A \otimes_{C} B = $}
\mbox{$ \overline{B} \cap {( \overline{B} \cap A^\bot \cap C )}^\bot \cap C = $}
\mbox{$ \overline{B} \cap {( \overline{B} \cap A^\bot )}^\bot = $}
\mbox{$ A \otimes B $}.
\item 
\mbox{$ A \oplus_{C} B = $} 
\mbox{$ {( B^{\bot_{C}} \otimes_{C} A^{\bot_{C}} )}^{\bot_{C}} $}.
But 
\[
B^{\bot_{C}} \otimes_{C} A^{\bot_{C}} = 
A^{\bot_{C}} \cap {( A^{\bot_{C}} \cap \overline{B} )}^{\bot_{C}} = 
A^\bot \cap C \cap {( A^\bot \cap C \cap \overline{B} )}^\bot \cap C = 
\]
\[
A^\bot \cap {( A^\bot \cap \overline{B} )}^\bot \cap C = 
( B^\bot \otimes A^\bot ) \cap C.
\] 
Therefore, by Lemma~\ref{the:1-dim}, item~\ref{pS}, \mbox{$ A \oplus_{C} B = $} 
\mbox{$ {( ( B^\bot \otimes A^\bot ) \cap C )}^{\bot_{C}} = $}
\mbox{$ \overline{ A \oplus B \cup C^\bot } \cap C = $}
\mbox{$ A \oplus B $}.
\item
Let us deal with property {\bf O}:
By {\bf O} in $X$, 
\mbox{$ x \in \overline{ \{ x \} \otimes A \cup \{ x \} \otimes A^\bot } $} and
\mbox{$ \{ x \} \otimes A \subseteq $} 
\mbox{$ \overline{ \{ x\} \cup \{ x \} \otimes A^\bot } $}.
By item~\ref{same-otimes} above, since \mbox{$ x \in C $} and \mbox{$ A \subseteq C $}, 
\mbox{$ \{ x \} \otimes A = $} \mbox{$ \{ x \} \otimes_{C} A $}
and therefore it is enough to show that \mbox{$ \{ x \} \otimes A^\bot = $} 
\mbox{$ \{ x \} \otimes_{C} ( A^\bot \cap C ) $}. Property {\bf O} for the subspace will 
then follow from item~\ref{same} above.
But \mbox{$ \{ x \} \otimes_{C} ( A^\bot \cap C ) = $} 
\mbox{$ \{ x \} \otimes ( A^\bot \cap C ) $}.
By Lemma~\ref{the:symmetric}, items~\ref{OAUB} and~\ref{union-bot}:
\[
\{ x \} \otimes ( A^\bot \cap C ) = 
\{ x \} \otimes ( A^\bot \cap C \cap \overline{ \{ x \} \cup {( A^\bot \cap C )}^\bot } ) = 
\{ x \} \otimes ( A^\bot \cap ( C \cap \overline{ \{ x \} \cup \overline{A} \cup C^\bot } ) ).
\]
Since \mbox{$ ( \{ x \} \cup A ) \bot \, C^\bot $}, Lemma~\ref{the:1-dim}, item~\ref{pS} and 
Lemma~\ref{the:symmetric}, item~\ref{OAUB} imply
\[
\{ x \} \otimes ( A^\bot \cap C ) = 
\{ x \} \otimes ( A^\bot \cap \overline{ \{ x \} \cup A } ) = \{ x \} \otimes A^\bot.
\]
\item
For property {\bf A}, let \mbox{$ A , B \in {\cal F_{C}} $}.
By {\bf A} in the original structure, we have \mbox{$ A \otimes B \subseteq $}
\mbox{$ \bigcup_{a \in A} \{ a \} \otimes B $}.
By item~\ref{same-otimes} above, \mbox{$ A \otimes_{C} B \subseteq $}
\mbox{$ \bigcup_{a \in A} \{ a \} \otimes_{C} B $}.
\end{enumerate}
\end{proof}

We proceed to prove the properties of O-spaces that are crucial for the deduction system
to be presented in Section~\ref{sec:rules}.
Our next result shows that, in some respects, $\otimes$ works as a conjunction and $\oplus$ 
as a disjunction.
Item~\ref{associative} below is particularly interesting.
It expresses that, if two measurements $A$ and $B$ commute, 
the order in which they are performed is unimportant.
This may sound as a tautology, but, for an operation such as $\otimes$ that is not associative,
as already noted in~\cite{Lehmann_andthen:JLC} it is, in fact, highly significant.
The importance of item~\ref{moving} will be seen in Lemma~\ref{the:passing}:
it justifies moving a proposition over the turnstile.
Item~\ref{T} will be used to prove the soundness of {\bf Right-$\wedge$}.
\begin{theorem} \label{the:various}
For any \mbox{$ A , B , C \in {\cal F} $},
\begin{enumerate}
\item \label{CBA}
if \mbox{$ A \subseteq B $}, then
\mbox{$ ( C \otimes B ) \otimes A =$} \mbox{$ C \otimes A = $} 
\mbox{$ ( C \otimes A ) \otimes B $},
\item \label{associative}
if \mbox{$ A \otimes B =$} \mbox{$ B \otimes A $}, then 
\mbox{$ ( C \otimes A ) \otimes B =$} \mbox{$ ( C \otimes B ) \otimes A $},
\item \label{T}
\mbox{$ ( A \otimes B ) \otimes C = Z $} iff 
\mbox{$ A \otimes ( C \otimes B ) = Z $},
\item \label{moving}
\mbox{$ ( A \otimes B^\bot ) \subseteq C $} iff \mbox{$ A \subseteq ( B \oplus C ) $}.
\end{enumerate}
\end{theorem}
\begin{proof}
\begin{enumerate}
\item 
Assume \mbox{$ A \subseteq B $}.
By Lemma~\ref{the:1-dim}, item~\ref{neg-o-axioms1}, 
\mbox{$ B \cap {( C \otimes B )}^\bot = $} \mbox{$ B \cap C^\bot $}. 
By intersecting both terms with $A$ we obtain 
\mbox{$ A \cap {( C \otimes B )}^\bot = $} \mbox{$ A \cap C^\bot $} 
and, by Equations~(\ref{eq:proj}), we conclude that 
\[
( C \otimes B ) \otimes A = A \cap {( A \cap {( C \otimes B )}^\bot )}^\bot = 
A \cap {( A \cap C^\bot )}^\bot = C \otimes A.
\]
But, we also have \mbox{$ C \otimes A \subseteq $} \mbox{$ A \subseteq B $}.
By Lemma~\ref{the:1-dim}, item~\ref{small-big}, 
\mbox{$ ( C \otimes A ) \otimes B = $} \mbox{$ C \otimes A $}.
\item 
Assume \mbox{$ A \otimes B = B \otimes A $}.
By Equations~(\ref{eq:proj}), \mbox{$ A \otimes B \subseteq $} 
\mbox{$ A \cap B $}.
By Lemma~\ref{the:symmetric}, item~\ref{AinterB}, \mbox{$ A \otimes B = $} 
\mbox{$ A \cap B $}.
By Equations~(\ref{eq:proj}) \mbox{$ C \otimes A \subseteq $} \mbox{$ A $} 
and by Lemma~\ref{the:1-dim}, item~\ref{left-additivity}, 
\mbox{$ ( C \otimes A ) \otimes B \subseteq $}
\mbox{$ A \otimes B = A \cap B $}.
By Lemma~\ref{the:1-dim}, item~\ref{small-big} and item \ref{CBA} just above:
\[
( C \otimes A ) \otimes B = 
( ( C \otimes A ) \otimes B ) \otimes ( A \cap B ) = 
( C \otimes A ) \otimes ( A \cap B ) = 
C \otimes ( A \cap B ).
\]
Similarly, one shows that \mbox{$ ( C \otimes B ) \otimes A = $}
\mbox{$ C \otimes ( B \cap A ) $}.
\item
By Lemma~\ref{the:commute} our claim is equivalent to:
\mbox{$ ( A \otimes B ) \, \bot \, C $} iff \mbox{$ ( C \otimes B ) \, \bot \, A $}.
But \mbox{$ B \cap {( B \cap A^\bot )}^\bot \subseteq C^\bot $} implies
\mbox{$ B \cap {( B \cap A^\bot )}^\bot \subseteq $} \mbox{$ B \cap C^\bot $}
and also \mbox{$ {( B \cap C^\bot )}^\bot \subseteq $}
\mbox{$ {( B \cap {( B \cap A^\bot )}^\bot )}^\bot $}
and therefore, by Lemma~\ref{the:1-dim}, item~\ref{pS}:
\[
B \cap {( B \cap C^\bot )}^\bot \subseteq 
B \cap {( B \cap {( B \cap A^\bot )}^\bot )}^\bot =
B \cap \overline{ B^\bot \cup ( B \cap A^\bot ) } = 
B \cap A^\bot \subseteq A^\bot
\]
Conversely, if \mbox{$ ( C \otimes B ) \, \bot \, A $}, by Lemma~\ref{the:1-dim}, 
item~\ref{neg-o-axioms1},
\mbox{$ A \otimes B = $} 
\mbox{$ B \cap {( B \cap A^\bot )}^\bot \subseteq $}
\mbox{$ B \cap {( B \cap ( C \otimes B ) )}^\bot = $}
\mbox{$ B \cap {( C \otimes B )}^\bot = $}
\mbox{$ B \cap C^\bot $}, implying \mbox{$ ( A \otimes B ) \, \bot \, C $}.
\item
\mbox{$ ( A \otimes B^\bot ) \subseteq C $} iff
\mbox{$ ( A \otimes B^\bot ) \, \bot \, C^\bot $} iff, by Lemma~\ref{the:commute},
\mbox{$ ( A \otimes B^\bot) \otimes C^\bot = Z $} iff, by item~\ref{T} just above,
\mbox{$ A \otimes ( C^\bot \otimes B^\bot ) = Z $} iff, 
by Lemma~\ref{the:commute}, \mbox{$ A \, \bot \, C^\bot \otimes B^\bot  $} iff
\mbox{$ A \subseteq $} \mbox{$ B \oplus C $}.
\end{enumerate}
\end{proof}

\section{Language and interpretation} \label{sec:syntax}
Now we want to define the formulas of our logic, its sequents~\cite{Gent:32, Gent:69} 
and their semantics.
We shall do that in two stages: first we define an extended language ($\cal L'$) 
and then we shall show that we can restrict our attention to a restricted language ($\cal L$) 
and restricted sequents that will be easier for us to study.
This section is fairly pedestrian: its only purpose is to prepare the ground 
for Section~\ref{sec:rules} and all the technical preparatory work has been done 
in Section~\ref{sec:O-properties}.
The only point worth noticing is our interpretation of the sequents: 
association to the left on the left of the turnstile and to the right on its right.

\subsection{Extended language} \label{sec:language-ext}
\subsubsection{Propositions and Sequents} \label{sec:propositions-ext}
We consider a set $AT$ of atomic propositions, 
one unary and two binary connectives.
We shall represent propositions by small greek letters.
\begin{definition} \label{def:propositions-ext}
\ 
\begin{itemize}
\item
An atomic proposition is a proposition.
\item
If $\alpha$ is a proposition then \mbox{$ \neg \alpha$} is a proposition.
\item
If $\alpha $ and $\beta$ are propositions, then \mbox{$ \alpha \wedge \beta $} is a proposition.
\item
If $\alpha $ and $\beta$ are propositions, then \mbox{$ \alpha \vee \beta $} is a proposition.
\item
There are no other propositions.
\end{itemize}
\end{definition}
The set of propositions on $AT$ will be denoted by $\cal L'$.

\begin{definition} \label{def:sequents}
A sequent is constituted by two finite sequences of propositions, 
separated by the turnstile symbol, that may be $\models$ or $\vdash$.
\end{definition} 
Here is a typical sequent:
\mbox{$ \alpha , \neg \beta \vee \gamma \models \neg \neg ( \delta \wedge \epsilon ) $}.
In the representation of sequents we shall use greek capital letters to represent finite sequences
of propositions.
The sequent \mbox{$ \Gamma \models \Gamma $}
is a sequent in which the same sequence $\Gamma$ of propositions appears on both sides 
of the turnstile.

\subsubsection{Interpretation} \label{sec:interpretation}
In an O-space \mbox{$ \langle X , \bot , {\cal F} \rangle $}, 
a proposition is interpreted as a member of $\cal F$ with the help of an assignment 
function for atomic propositions.
\begin{definition} \label{def:interpretation}
Given an O-space \mbox{$ \langle X , \bot , {\cal F}\rangle $}.
A consequence of property {\bf F} in Definition~\ref{def:O-space} is that 
any assignment \mbox{$ v : AT \longrightarrow {\cal F} $} of elements of $\cal F$
to the atomic propositions can be extended to a function 
\mbox{$ v : {\cal L'} \longrightarrow {\cal F } $} by
\begin{itemize}
\item
\mbox{$ v( \neg \alpha ) = v(\alpha)^\bot $} for any \mbox{$\alpha \in {\cal L'}$},
\item
\mbox{$ v( \alpha \wedge \beta ) = v(\alpha) \otimes v(\beta) $} for any 
\mbox{$\alpha , \beta \in {\cal L'}$},
\item
\mbox{$ v( \alpha \vee \beta ) = v(\alpha) \oplus v(\beta) $} for any 
\mbox{$\alpha , \beta \in {\cal L'}$}.
\end{itemize}
\end{definition}
\begin{definition} \label{def:equivalence}
For any \mbox{$ \alpha , \beta \in {\cal L'} $} we shall say that $\alpha$ and $\beta$
are {\em semantically equivalent} and write \mbox{$ \alpha \equiv \beta $} iff 
\mbox{$ v(\alpha) = v(\beta) $} for any O-space and any assignment of flats to the atomic 
propositions.
\end{definition}

\begin{lemma} \label{the:equivalence}
For any \mbox{$ \alpha , \beta \in {\cal L'} $}
\begin{itemize}
\item
\mbox{$ \neg ( \alpha \wedge \beta ) \equiv \neg \beta \vee \neg \alpha $},
\item
\mbox{$ \neg ( \alpha \vee \beta ) \equiv \neg \beta \wedge \neg \alpha $},
\item \label{double-negation}
\mbox{$ \neg ( \neg \alpha ) \equiv \alpha $}.
\end{itemize}
\end{lemma}
\pagebreak[2]
\begin{proof}
\begin{itemize}
\item
By Lemma~\ref{the:symmetric}, item~\ref{overline} and the fact that 
\mbox{$ v(\beta)^\bot \otimes v(\alpha^\bot) $} is a flat.
\[
v( \neg ( \alpha \wedge \beta ) ) =  {( v(\alpha) \otimes v(\beta))}^\bot =
{ {( v(\beta)^\bot \oplus v(\alpha)^\bot )}^\bot }^\bot =
v(\beta)^\bot \oplus v(\alpha)^\bot  = v ( \neg \beta \vee \neg \alpha ).
\]
\item
Similarly.
\item
\mbox{$ v ( \neg ( \neg \alpha ) ) = $}
\mbox{$ { v( \alpha)^\bot}^\bot =$}
\mbox{$ v( \alpha) $} since $v(\alpha)$ is a flat.
\end{itemize}
\end{proof}
We must now interpret sequents. 
We shall, as expected, interpret the left-hand side and the right-hand side as flats and
the turnstile $\models$ as set inclusion.
In the literature, following Gentzen~\cite{Gent:32} the comma on the left-hand side is interpreted as a conjunction and the
comma on the right-hand side as a disjunction.
Theorem~\ref{the:various} shows that $\wedge$ is some kind of conjunction and $\vee$ some
kind of disjunction, it is therefore natural to interpret the comma on the left-hand side as
$\wedge$ and the comma on the right-hand side as $\vee$.
Since those operations are not associative, we must decide how to associate the elements
of the left-hand side and how to associate those of the right-hand side.
Since \mbox{$A \wedge B$} denotes the result of measuring $B$ \emph{after} $A$, it is
natural to decide that the elements of the left-hand side associate to the left.
To keep the action close to the turnstile, we decide that the elements of the right-hand side  
associate to the right.

Our interpretation, in an O-space \mbox{$\langle X , \bot , {\cal F} \rangle $}, 
for assignment \mbox{$ v: {\cal L} \rightarrow {\cal F} $}, of sequent:
\begin{equation} \label{eq:sequent}
\alpha_{0} , \alpha_{1} \ldots \alpha_{n - 1} \models \beta_{0} , \beta_{1}  \ldots 
\beta_{m - 1}
\end{equation}
is therefore
\begin{equation} \label{eq:seq-interpretation}
v ( ( \ldots ( \alpha_{0} \wedge \alpha_{1} ) \wedge \ldots ) \wedge \alpha_{n - 1} ) 
\subseteq v ( \beta_{0} \vee ( \beta_{1} \vee ( \dots \vee \beta_{m - 1}) \ldots ) )
\end{equation}
If the right-hand side of the turnstile is empty its interpretation is $Z$.
If the left-hand side of the turnstile is empty its interpretation is $X$.

\begin{definition} \label{def:valid-sequent}
A sequent is \emph{valid in an O-space \mbox{$ \langle X , \bot , {\cal F} \rangle $} }
iff its interpretation holds for every assignment.
It is \emph{valid} iff it is valid in every O-space.
\end{definition}

For example, the sequent \mbox{$ \alpha \models \alpha , \alpha $} is valid 
for any $\alpha$ since \mbox{$ A = A \oplus A $} for any flat $A$ 
by Lemma~\ref{the:symmetric}, item~\ref{AsubB1}.
Our next lemma shows that propositions can jump over the turnstile in both direction at the cost
of an added negation. 
\begin{lemma} \label{the:passing}
For any O-space \mbox{$ O = \langle X , \bot , {\cal F} \rangle $}, 
\mbox{$ \alpha_{0} , \ldots , \alpha_{n} \models \beta_{0} , \beta_{1} , \ldots , \beta_{m} $} 
is valid in $O$ iff 
\mbox{$  \alpha_{0} , \ldots , \alpha_{n} , \neg \beta_{0} \models 
\beta_{1} , \ldots , \beta_{m} $} is valid in $O$.
Similarly,
\mbox{$ \alpha_{0} , \ldots , \alpha_{n} \models \beta_{0} , \beta_{1} , \ldots , \beta_{m} $} 
is valid in $O$ iff 
\mbox{$  \alpha_{0} , \ldots , \alpha_{n - 1} \models \neg \alpha_{n} , \beta_{0} ,
\beta_{1} , \ldots , \beta_{m} $} is valid in $O$.
\end{lemma}
\begin{proof}
By Theorem~\ref{the:various}, item~\ref{moving} and the fact that 
\mbox{$ \neg \neg \alpha \equiv \alpha $}.
\end{proof}
The nature of the quantum negation has been discussed in the literature, e.g., 
in~\cite{Connectives:Fest}.
On one hand since orthogonal ($\bot$) is stronger than distinct ($\neq$) one may be tempted 
to conclude that quantum negation ($\neg$) is a strong negation, possibly akin to
an intuitionistic negation. 
Lemma~\ref{the:passing} above shows that this is not so: \emph{quantum negation is classical}.

\subsection{Restricted language} \label{sec:language-restricted}
In view of Lemma~\ref{the:equivalence}, we can restrict the set of propositions we are 
interested in by pushing all negations inside the connectives $\wedge$ and $\vee$, and,
since double negations can be eliminated, we can assume that negations apply only
to atomic propositions.

\begin{definition} \label{def:propositions-restricted}
The set $\cal L$ of propositions is defined in the following way.
\begin{itemize}
\item
An atomic proposition is a proposition.
\item
If $\alpha$ is an \emph{atomic} proposition then \mbox{$ \neg \alpha$} is a proposition.
\item
If $\alpha $ and $\beta$ are propositions, then \mbox{$ \alpha \wedge \beta $} is a proposition.
\item
If $\alpha $ and $\beta$ are propositions, then \mbox{$ \alpha \vee \beta $} is a proposition.
\item
There are no other propositions.
\end{itemize}
\end{definition}

In view of Lemma~\ref{the:passing}, we can also restrict our attention to sequents 
with an empty right-hand side, sequents of the form
\mbox{$ \alpha_{0} , \ldots , \alpha_{n} \models $}, whose intuitive meaning is:
the sequence of measurements $\alpha_{0}$, $\alpha_{1}$, \ldots $\alpha_{n}$ 
\emph{in this order} will never be observed.

\section{Deduction rules} \label{sec:rules}
We shall describe a system $\cal R$ of ten deduction rules 
and prove it is sound and complete for the logic of O-spaces.
In Section~\ref{sec:soundness} the rules are described and proved to be valid:
a classical Cut Rule, an Exchange rule limited to a sequence of two propositions, 
three different Weakening rules, a logical axiom that is the negation introduction rule,
two $\otimes$ introduction rules and two $\oplus$ introduction rules.
The system provides many symmetries and is aesthetically pleasant.
The question whether a Cut elimination theorem can be proved is open.
The rules are not claimed to be independent.
Section~\ref{sec:derived} proves the validity of a number of derived rules, in particular some
elimination rules.
Section~\ref{sec:completeness} proves the deductive system is complete: there is some
O-space for which the valid sequents are exactly the sequents provable in the deductive system.

\subsection{A sound deductive system} \label{sec:soundness}
In deduction rules we use the symbol $\vdash$ to separate the left side 
from the right side of a sequent and not $\models$ as above.
A deduction rule consists of a finite set of sequents, \emph{the assumptions} 
and a sequent, \emph{the conclusion} separated by a horizontal line, called the inference line.
For example, consider the following Cut rule, which is the classical Cut rule.

\( \begin{array} {lc} \\
{R1  \ \bf Cut} &
\begin{array}{c}
\Gamma , \alpha , \Delta \vdash \ \ \ \ \Gamma , \neg \alpha , \Delta \vdash \\
\hline
\Gamma , \Delta \vdash
\end{array} 
\end{array} \)

Such a rule is meant to be part of a set of deduction rules and its meaning is: 
if one has already established the two sequents above the inference line, then, one is
entitled to establish the sequent below the inference line.
We are interested in two properties of such rules and systems of rules.
\begin{definition} \label{def:sound-complete}
An deduction rule is said to be \emph{sound} iff, for any O-space and any assignment $v$, 
for which all the assumptions are valid, the conclusions are also valid 
(in the specified O-space, with the specified assignment $v$).
A set of deduction rules is \emph{complete} iff any sequent that is \emph{valid} 
(in all O-spaces, for all assignments) can be derived using only the rules in the set.
\end{definition}

Let's see that {\bf Cut} is indeed sound.
For any sequence of propositions: 
\mbox{$ \Sigma = $} \mbox{$ \sigma_{0} , \sigma_{1} , \ldots , \sigma_{n} $}
let us denote by \mbox{$ \phi(\Sigma) $} the proposition 
\mbox{$ ( \ldots ( \sigma_{0} \wedge \sigma_{1} ) \wedge \ldots ) \wedge \sigma_{n} $}.
Our claim follows from Lemma~\ref{the:widehat} below and the fact that the assumptions 
of the {\bf Cut} rule imply that \mbox{$ v ( \phi( \Gamma , \alpha , \Delta ) ) =$}
\mbox{$ Z =$} \mbox{$ v ( \phi ( \Gamma , \neg \alpha , \Delta ) ) $}.

\begin{lemma} \label{the:widehat}
For any sequences $\Gamma$, $\Delta$ and any proposition $\alpha$,
\[
v ( \phi ( \Gamma , \Delta ) ) \subseteq 
\overline{ v ( \phi ( \Gamma , \alpha , \Delta ) ) \cup 
v ( \phi ( \Gamma , \neg \alpha , \Delta ) ) }.
\]
\end{lemma}
\begin{proof}
We reason by induction on the size of the sequence $\Delta$.
If $\Delta$ is empty, then the claim follows directly from {\bf O} and 
Lemma~\ref{the:symmetric}, item~\ref{bar-cup}:
for any \mbox{$ A , B \in {\cal F} $}, \mbox{$ B \subseteq $}
\mbox{$ \overline{ B \otimes A \cup B \otimes A^\bot } $}.
Assume we have proved the result for $\Delta$ and let us prove it for the extended 
sequence \mbox{$ \Delta , \delta $}.
By Lemma~\ref{the:1-dim}, item~\ref{left-additivity} and the induction hypothesis
\[
v ( \phi ( \Gamma , \Delta , \delta ) ) =
v ( \phi ( \Gamma , \Delta ) \wedge \delta ) \subseteq
\overline { v ( \phi ( \Gamma , \alpha , \Delta ) \wedge \delta ) \cup
v ( \phi ( \Gamma , \neg \alpha , \Delta ) \wedge \delta ) } =
\]
\[
\overline { v ( \phi ( \Gamma , \alpha , \Delta , \delta ) ) \cup 
v ( \phi ( \Gamma , \neg \alpha , \Delta , \delta ) ) }.
\]
\end{proof}
Let us, first, consider structural rules, i.e., rules that do not involve the connectives.
There is no valid general {\bf Exchange} rule: one cannot modify the order of the propositions 
in the left-hand side of a sequent, but there is a very limited exchange rule: if a sequent
has only two propositions in its left part, those two propositions may be exchanged.

\( \begin{array} {lc} \\
{R2 \ \bf Exchange} &
\begin{array}{c}
\alpha , \beta \vdash \\
\hline
\beta , \alpha \vdash
\end{array} 
\end{array} \)

The soundness of {\bf Exchange} follows from Lemma~\ref{the:commute}
and the fact that \mbox{$ v(\alpha) \, \bot \, v(\beta) $} iff 
\mbox{$ v(\beta) \, \bot \, v(\alpha) $}.

Concerning {\bf Weakening} the situation is more complex: one cannot add a formula
anywhere in a sequent, but one can do this in three different situations:
\begin{itemize}
\item
in the leftmost position:

\( \begin{array} {lc} \\
{R3 \ \bf Left-Weakening} &
\begin{array}{c}
\Gamma \vdash \\
\hline
\alpha , \Gamma \vdash
\end{array} 
\end{array} \)

\item
in the rightmost position:

\( \begin{array} {lc} \\
{R4 \ \bf Right-Weakening} &
\begin{array}{c}
\Gamma \vdash \\
\hline
\Gamma , \alpha \vdash
\end{array} 
\end{array} \)
\item
one can insert a formula in any position, as long as it already holds at this point.
This is an introduction rule for negation.

\( \begin{array} {lc} \\
{R5 \ \bf Stuttering} &
\begin{array}{c}
\Gamma , \Delta \vdash \ \ \ \ \Gamma , \alpha \vdash \\
\hline
\Gamma , \neg \alpha , \Delta \vdash
\end{array} 
\end{array} \)
\end{itemize}

The soundness of {\bf Left-Weakening} follows from the following remark.
Let \mbox{$ O = \langle X , \bot , {\cal F} \rangle $} be an O-space.
Let \mbox{$ \{ A_{i} \}, 0 \leq i \leq n - 1 $} be any sequence of elements of $\cal F$.
If \mbox{$ x_{n} \in ( \ldots ( A_{0} \otimes A_{1} ) \otimes \ldots ) \otimes A_{n} $},
then, by Lemma~\ref{the:1-dim}, item~\ref{left-additivity}, there is some 
\mbox{$ x_{n - 1} \in \overline{A_{n - 1 }} $} such that 
\mbox{$ x_{n} \in \{ x_{n - 1} \} \otimes A_{n} $}.
Similarly, for any $i$, \mbox{$ 0 \leq  i \leq n $}, there is some 
\mbox{$ x_{i - 1} \in \overline{A_{i  - 1} } $} such that 
\mbox{$ x_{i} \in \{ x_{i - 1} \} \otimes A_{i} $}.
Therefore 
\mbox{$ x_{n} \in \{ x_{0} \} \otimes ( \ldots ( A_{1} \otimes A_{2} ) \ldots ) 
\otimes A_{n} $} by Lemma~\ref{the:1-dim}, item~\ref{left-additivity} and
\mbox{$ ( \ldots ( A_{0} \otimes A_{1} ) \otimes \ldots ) \otimes A_{n} \subseteq $}
\mbox{$ ( \ldots ( A_{1} \otimes A_{2} ) \otimes \ldots ) \otimes A_{n} $}.
If the right-hand side is equal to $Z$, so is the left-hand side.

The {\bf Right-Weakening} rule is seen to be sound since 
\mbox{$ Z \otimes A = Z $} by Equations~(\ref{eq:proj}) and 
Lemma~\ref{the:symmetric}, item~\ref{empty-full}.

The {\bf Stuttering} rule is sound since, for any \mbox{$ A , B \in {\cal F} $} such that 
\mbox{$ A \otimes B = Z $} one has \mbox{$ A \otimes B^\bot = A $}.
Indeed, by Lemma~\ref{the:commute}, \mbox{$ A \otimes B = Z $} implies
\mbox{$ A \, \bot \, B $} and \mbox{$ A \subseteq B^\bot $}.
Therefore, by Lemma~\ref{the:1-dim}, item~\ref{small-big},
\mbox{$ A \otimes B^\bot = A $}.

Let us proceed, now with rules that deal with each of our three connectives.
All the rules of our system are introduction rules: the propositions that appear in the assumptions 
also appear in the conclusion.
First, a rule for negation.
It is an axiom, i.e., a deduction rule with no assumptions.
For any \emph{atomic} proposition $\sigma$:

\( \begin{array} {lc} \\
{R6 \ \bf \neg Atomic} &
\begin{array}{c}
\\
\hline
\sigma , \neg \sigma \vdash
\end{array} 
\end{array} \)

It is sound, since, in any O-space, for any element $A$ of $\cal F$ one has
\mbox{$ A \, \bot \, A^\bot $} and, by Lemma~\ref{the:commute}, 
\mbox{$ A \otimes A^\bot = Z $} and therefore
for any proposition $\alpha$ one has \mbox{$ \alpha , \neg \alpha \models $}.
Note that we have shown that the rule is sound for any proposition $\alpha$, but we restrict
the rule to atomic propositions because we prefer to work with a weaker set of deduction rules
in order to prove stronger results.
Let us now consider rules for the connectives $\wedge$ and $\vee$.
For $\wedge$ we have two deduction rules, one on the leftmost part of the sequent and one
on the rightmost part.

\(\begin{array}{lc} \\
{R7 \ \bf Left-\wedge} &
\begin{array}{c}
\alpha , \beta , \Delta \vdash \\
\hline
\alpha \wedge \beta , \Delta \vdash
\end{array} 
\end{array} \)

\(\begin{array}{lc} \\
{R8 \ \bf Right-\wedge}  &
\begin{array}{c}
\Gamma , \alpha , \beta \vdash \\
\hline
\Gamma , \beta \wedge \alpha \vdash
\end{array} 
\end{array} \)

Note the change of order between $\alpha$ and $\beta$ in {\bf Right-$\wedge$}.
The soundness of {\bf Left-$\wedge$} follows directly from our interpretation of 
sequents in Equation~(\ref{eq:seq-interpretation}).
The soundness of {\bf Right-$\wedge$} follows from Theorem~\ref{the:various}, 
item~\ref{T}.
Similarly, we have two deduction rules for the connective $\vee$.
The first one is an axiom: for any propositions \mbox{$\alpha , \beta$}:

\( \begin{array} {lc} \\
{R9 \ \bf \neg\vee1} &
\begin{array}{c}
\\
\hline
\neg \alpha , \neg \beta , \, \beta \vee \alpha \vdash 
\end{array} 
\end{array} \)

The soundness of {\bf $\neg\vee$1} follows from the fact that, in any O-space, for any 
\mbox{$ A , B \in {\cal F} $} one has 
\mbox{$ A \otimes B \, \perp \, \neg ( A \otimes B ) $},
and by Lemma~\ref{the:commute}, 
\mbox{$ ( A \otimes B ) \otimes ( B^\bot \oplus A^\bot ) = Z $}.
Therefore, for any propositions $\alpha$ and $\beta$, we have
\mbox{$ \alpha , \beta , \neg \beta \vee \neg \alpha \models $}.

\(\begin{array}{lc} \\
{R10 \ \bf \vee-Intro} &
\begin{array}{c}
\Gamma , \alpha , \Delta \vdash \ \ \ \ \ \Gamma , \neg \alpha , \beta , \Delta \vdash \\
\hline
\Gamma , \alpha \vee \beta , \Delta \vdash
\end{array} 
\end{array} \)

The soundness of {\bf $\vee$-Intro} is proved by the following lemma.
\begin{lemma} \label{the:Left-oplus}
For any sequence $\Delta$ and any propositions $\alpha$, $\beta$, $\gamma$,
\[
v ( \phi ( \gamma \wedge ( \alpha \vee \beta ) , \Delta ) ) \subseteq 
\overline{ v ( \phi ( \gamma \wedge \alpha , \Delta ) ) \cup 
v ( \phi ( \gamma \wedge \neg \alpha , \beta , \Delta ) ) }.
\]
\end{lemma}
\begin{proof}
By induction on the size of the sequence $\Delta$.
Assume, first that the sequence $\Delta$ is empty.
We want to prove that, for any \mbox{$ A , B , C \in {\cal F} $} such that 
\mbox{$ C \otimes A = Z $} and \mbox{$ ( C \otimes A^\bot ) \otimes B = Z $} 
one has \mbox{$ C \otimes ( A \oplus B ) = Z $}.
By Theorem~\ref{the:various}, item~\ref{T}, the second assumption implies that
\mbox{$ C \otimes ( B \otimes A^\bot ) = Z $}.
By Lemma~\ref{the:commute}, then, the assumptions imply that \mbox{$ C \, \bot \, A $}
i.e., \mbox{$ C \subseteq A^\bot $}  and
\mbox{$ C \, \bot \, B \otimes A^\bot $}, i.e., 
\mbox{$ C \subseteq {( B \otimes A^\bot )}^\bot $}.
We see that \mbox{$ C \subseteq $}
\mbox{$ A^\bot \cap {( B \otimes A^\bot )}^\bot \subseteq $}
\mbox{$ A^\bot \cap {( B \cap A^\bot )}^\bot =$}
\mbox{$ B^\bot \otimes A^\bot $}.
One concludes that \mbox{$ C \, \bot \, A \oplus B $}.
We have proved the soundness of the rule when $\Delta$ is empty.
For the induction step, by Lemma~\ref{the:1-dim}, item~\ref{left-additivity}, 
and the induction hypothesis:
\[
v ( \phi ( \gamma\wedge ( \alpha \vee \beta ) , \Delta ) \wedge \delta ) =
v ( \delta ) \cap 
{( v ( \delta ) \cap {( v ( \phi ( \gamma \wedge ( \alpha \vee \beta ) , \Delta ) ) )}^\bot )}^\bot \subseteq 
\]
\[
v ( \delta ) \cap 
{( v ( \delta ) \cap {( v ( \phi ( \gamma \wedge \alpha , \Delta ) ) \cup 
v ( \phi ( \gamma \wedge \neg \alpha , \beta , \Delta ) ) )}^\bot )}^\bot = 
\]
\[
( v ( \phi ( \gamma \wedge \alpha , \Delta ) ) \cup 
v ( \phi ( \gamma \wedge \neg \alpha , \beta , \Delta ) ) ) \otimes \delta =
( v ( \phi ( \gamma \wedge \alpha , \Delta ) ) ) \otimes \delta \cup 
( v ( \phi ( \gamma \wedge \neg \alpha , \beta , \Delta ) ) ) \otimes \delta = 
\]
\[
v ( \phi ( \gamma \wedge \alpha , \Delta , \delta ) ) \cup 
v ( \phi ( \gamma \wedge \neg \alpha , \beta , \Delta , \delta ) ). 
\]
\end{proof}

It is worth noticing that, even though quantum logic is not bivalent as explained 
in Section~\ref{sec:O-def}, the sequent \mbox{$ \vdash \alpha \vee \neg \alpha $} 
is valid, i.e., quantum logic satisfies the Law of Excluded Middle.
One also easily sees that if the standard unlimited Exchange rule is added to our system, 
one obtains classical propositional logic where $\wedge$ is conjunction and 
$\vee$ is disjunction. 
We can now state a soundness theorem: its proof has been provided above.
\begin{theorem} \label{the:soundness}
Each of the ten deductive rules of the system $\cal R$ described above is sound.
\end{theorem}

\subsection{Derived rules} \label{sec:derived}
To prepare the completeness result of Theorem~\ref{the:completeness} we need
to put in evidence the power of the deductive system presented above.
The rules that we shall derive are the counterparts of the properties of O-spaces 
mentioned in Section~\ref{sec:O-properties}.
By Theorem~\ref{the:soundness}, the rules presented below are sound, but 
our purpose is to show that they can be derived, \emph{not that they are sound}.

\subsubsection{A first batch of derived rules} \label{sec:first-batch}
\begin{itemize}
\item
First, a generalization of the {\bf $\neg$ Atomic} rule.
For \emph{any} proposition $\alpha$:

\( \begin{array} {lc} \\
{\bf Logical \ Axiom} &
\begin{array}{c}
\\
\hline
\alpha , \neg \alpha \vdash
\end{array} 
\end{array} \)

Derivation:
If \mbox{$ \alpha = \sigma $} is atomic, the rule is {\bf $\neg$Atomic}.
If \mbox{$ \alpha = \neg \sigma $}, the negation of an atomic proposition:

\(
\begin{array}{lc}
{\bf \neg Atomic} & \sigma , \neg \sigma \vdash \\
\cline{2-2}
{\bf Exchange} & \neg \sigma , \sigma \vdash \\
\cline{2-2}
& \neg \sigma , \neg \neg \sigma \vdash 
\end{array}
\)

If \mbox{$ \alpha = \beta \wedge \gamma $}:

\(
\begin{array}{lc}
{\bf \neg\vee1}  & \beta , \gamma , \neg \beta \vee \neg \gamma \vdash \\
\cline{2-2}
{\bf Left-\wedge} & \beta \wedge \gamma , \neg \beta \vee \neg \gamma \vdash \\
\cline{2-2}
& \beta \wedge \gamma , \neg ( \gamma \wedge \beta ) \vdash
\end{array}
\)

If \mbox{$ \alpha = \beta \vee \gamma $}:

\(
\begin{array}{lc}
{\bf \neg\vee1}  & \neg \gamma , \neg \beta , \beta \vee \gamma \vdash \\
\cline{2-2}
{\bf Left-\wedge} & \neg \gamma \wedge \neg \beta , \beta \vee \gamma \vdash \\
\cline{2-2}
{\bf Exchange} & \beta \vee \gamma , \neg \gamma \wedge \neg \beta \vdash \\
\cline{2-2}
& \beta \vee \gamma , \neg ( \beta \vee \gamma ) \vdash
\end{array}
\)
\item
Now, a repetition rule:

\( \begin{array} {lc} \\
{\bf Repetition} &
\begin{array}{c}
\Gamma , \alpha , \Delta \vdash \\
\hline
\Gamma , \alpha , \alpha , \Delta \vdash
\end{array} 
\end{array} \)

Derivation:

\(
\begin{array}{lcrc}
{\bf Logical Axiom} & \alpha , \neg \alpha \vdash & & \\
\cline{2-2}
{\bf Left-Weakening} & \Gamma , \alpha , \neg \alpha \vdash 
& {\bf Assumption} & \Gamma , \alpha , \Delta \vdash \\
\cline{2-4}
& {\bf Stuttering} & \Gamma , \alpha , \alpha , \Delta \vdash
\end{array}
\)
\item
A contraction rule:

\( \begin{array} {lc} \\
{\bf Contraction} &
\begin{array}{c}
\Gamma , \alpha , \alpha , \Delta \vdash \\
\hline
\Gamma , \alpha , \Delta \vdash
\end{array} 
\end{array} \)

Derivation:

\(
\begin{array}{lcrc}
{\bf Logical Axiom} & \alpha , \neg \alpha \vdash \\
\cline{2-2}
{\bf Left-Weakening} & \Gamma , \alpha , \neg \alpha \vdash \\
\cline{2-2}
{\bf Right-Weakening} & \Gamma , \alpha , \neg \alpha , \Delta \vdash  
& {\bf Assumption} & \Gamma , \alpha , \alpha , \Delta \vdash \\
\cline{2-4}
& {\bf Cut} & \Gamma , \alpha , \Delta \vdash
\end{array}
\)
\item
Now, some elimination rules.
First, the converse to {\bf Left-$\wedge$}

\( \begin{array} {lc} \\
{\bf \wedge Left-Elim} &
\begin{array}{c}
\alpha \wedge \beta , \Delta \vdash \\
\hline
\alpha , \beta , \Delta \vdash
\end{array} 
\end{array} \)

Derivation:

\(
\begin{array}{lcrc}
{\bf Assumption} & \alpha \wedge \beta , \Delta \vdash 
& {\bf \neg\vee1} & \alpha , \beta , \neg ( \alpha \wedge \beta ) \vdash \\
\cline{2-2} \cline{4-4}
{\bf Left-Weakening} & \alpha , \beta , \alpha \wedge \beta , \Delta \vdash 
& {\bf Right-Weakening} & \alpha , \beta , \neg ( \alpha \wedge \beta ) , \Delta \vdash \\
\cline{2-4}
& {\bf Cut} & \alpha , \beta , \Delta \vdash
\end{array}
\)
\item
Some exchange rule among three propositions

\( \begin{array} {lc} \\
{\bf Circ} &
\begin{array}{c}
\alpha , \beta , \gamma \vdash \\
\hline
\gamma , \beta , \alpha \vdash 
\end{array} 
\end{array} \)

Derivation:

\(
\begin{array}{lc}
{\bf Assumption} & \alpha , \beta , \gamma \vdash \\
\cline{2-2}
{\bf Right-\wedge} & \alpha , \gamma \wedge \beta \vdash \\
\cline{2-2}
{\bf Exchange} & \gamma \wedge \beta , \alpha \vdash \\
\cline{2-2}
{\bf \wedge Left-Elim} & \gamma , \beta , \alpha \vdash
\end{array}
\)
\item
The converse to {\bf Right-$\wedge$}

\( \begin{array} {lc} \\
{\bf \wedge Right-Elim} &
\begin{array}{c}
\Gamma , \alpha \wedge \beta \vdash \\
\hline
\Gamma , \beta , \alpha \vdash
\end{array} 
\end{array} \)

Derivation:
\nopagebreak

\(
\begin{array}{lcrc}
& & {\bf \neg\vee1} & \alpha , \beta , \neg \beta \vee \neg \alpha \vdash \\
\cline{4-4}
{\bf Assumption} & \Gamma , \alpha \wedge \beta \vdash 
& {\bf Circ} & \neg \beta \vee \neg \alpha , \beta , \alpha \vdash \\
\cline{2-2} \cline{4-4}
{\bf Right-Weakening} & \Gamma , \alpha \wedge \beta , \beta , \alpha \vdash 
& {\bf Left-Weakening}  & \Gamma , \neg ( \alpha \wedge \beta ) , \beta , \alpha \vdash \\
\cline{2-4}
& {\bf Cut} & \Gamma , \beta , \alpha \vdash
\end{array}
\)
\item
An elimination rule for $\vee$.

\( \begin{array} {lc} \\
{\bf \vee Left-Elim} &
\begin{array}{c}
\alpha \vee \beta , \Delta \vdash \\
\hline
\alpha , \Delta \vdash 
\end{array}
\end{array} \)

Derivation:

\( \begin{array}{lcrc}
& & {\bf Logical Axiom} & \alpha , \neg \alpha \vdash \\
\cline{4-4}
& & {\bf Right-Weakening} & \alpha , \neg \alpha , \neg \beta \vdash \\
\cline{4-4}
{\bf Assumption}  & \neg ( \neg \beta \wedge \neg \alpha ) , \Delta \vdash
& {\bf Right-\wedge} &  \alpha , \neg \beta \wedge \neg \alpha \vdash \\
\cline{2-2} \cline{4-4}
{\bf Left-Weakening} & \alpha , \neg ( \neg \beta \wedge \neg \alpha ) , \Delta \vdash 
& {\bf Right-Weakening} & \alpha , \neg \beta \wedge \neg \alpha , \Delta \vdash \\
\cline{2-4}
& {\bf Cut} & \alpha , \Delta \vdash 
\end{array} \)
\item
Another one.

\( \begin{array} {lc} \\
{\bf \vee Right-Elim} &
\begin{array}{c}
\Gamma , \alpha \vee \beta \vdash \\
\hline
\Gamma , \alpha \vdash 
\end{array}
\end{array} \)

Derivation:

\( \begin{array}{lcrc}
& & {\bf Logical Axiom} & \alpha , \neg \alpha \vdash \\
\cline{4-4}
& & {\bf Exchange} & \neg \alpha , \alpha \vdash \\
\cline{4-4}
& & {\bf Left-Weakening} & \neg \beta , \neg \alpha , \alpha \vdash \\
\cline{4-4}
{\bf Assumption}  & \Gamma , \neg ( \neg \beta \wedge \neg \alpha ) \vdash
& {\bf Left-\wedge} &  \neg \beta \wedge \neg \alpha , \alpha \vdash \\
\cline{2-2} \cline{4-4}
{\bf Right-Weakening} & \Gamma, \neg ( \neg \beta \wedge \neg \alpha ) , \alpha \vdash 
& {\bf Left-Weakening} & \Gamma , \neg \beta \wedge \neg \alpha , \alpha \vdash \\
\cline{2-4}
& {\bf Cut} & \Gamma , \alpha \vdash 
\end{array} \)
\end{itemize}

\subsubsection{Implication and logical equivalence} \label{sec:implication}
We shall begin, now, an in-depth study of the deductive system $\cal R$. 
\begin{definition} \label{def:implication}
Let \mbox{$ \alpha , \beta \in {\cal L} $}. We shall say that $\beta$ \emph{implies} $\alpha$ 
and write \mbox{$ \beta \rightarrow \alpha $}
iff the sequent \mbox{$ \beta , \neg \alpha \vdash $} is derivable in the system $\cal R$.
We shall say that $\alpha$ and $\beta$ are \emph{logically equivalent} and write 
\mbox{$ \alpha \simeq \beta $} iff \mbox{$ \alpha \rightarrow \beta $} and 
\mbox{$ \beta \rightarrow \alpha $}.
\end{definition}

The following holds as expected and confirms the classical flavor of our system.
\begin{lemma} \label{the:implication}
\ 
\begin{enumerate}
\item \label{ref-trans}
The implication relation is reflexive and transitive,
\item \label{equivalence-relation}
logical equivalence is an equivalence relation, 
\item \label{equiv-sequent}
if \mbox{$ \alpha \simeq \beta $}, then,for any sequences $\Gamma$ and $\Delta$, 
one has \mbox{$ \Gamma , \alpha , \Delta \vdash $} iff 
\mbox{$ \Gamma , \beta , \Delta \vdash $},
\item \label{congruence}
logical equivalence is a congruence for $\neg$, $\wedge$ and $\vee$,
\item \label{imp}
if \mbox{$ \beta \rightarrow \alpha $}, \mbox{$ \neg \alpha \rightarrow \neg \beta $}
and for any \mbox{$ \gamma \in {\cal L} $}, 
\mbox{$ \beta \wedge \gamma \rightarrow \alpha \wedge \gamma $} and
\mbox{$ \gamma \vee \beta \rightarrow \gamma \vee \alpha $}.
\end{enumerate}
\end{lemma}
\begin{proof}
\begin{enumerate}
\item
Reflexivity follows from {\bf Logical Axiom}.
The following derivation proves transitivity:

\( \begin{array}{rlrl}
{\bf Assumption} & \neg \beta , \alpha \vdash 
& {\bf Assumption} & \neg \gamma , \beta \vdash \\
\cline{2-2} \cline{4-4} 
{\bf Left-Weakening} & \neg \gamma , \neg \beta , \alpha \vdash & 
{\bf Right-Weakening} & \neg \gamma , \beta , \alpha \vdash \\
\cline{2-4}
& {\bf Cut} & \neg \gamma , \alpha \vdash
\end{array} \)
\item
Obvious.
\item
Here is the derivation:

\( \begin{array}{rrrrrl}
{\bf Assumption} & \alpha , \neg \beta \vdash & &
& {\bf Assumption} &  \beta , \neg \alpha \vdash \\
\cline{2-2} \cline{6-6}
{\bf LWeakening} & \Gamma , \alpha , \neg \beta \vdash 
& {\bf Assumption} & \Gamma , \alpha , \Delta \vdash 
&  {\bf Exchange} &  \neg \alpha , \beta \vdash \\
\cline{2-4} \cline{6-6}
& {\bf Stuttering} & \Gamma , \alpha , \beta , \Delta \vdash & 
& {\bf LRWeakening} & \Gamma , \neg \alpha , \beta , \Delta \vdash \\
\cline{3-6}
& & {\bf Cut} & \Gamma , \beta , \Delta \vdash
\end{array} \)
\item
Assume that \mbox{$ \alpha \simeq \beta $}.
By item~\ref{equiv-sequent}, \mbox{$ \alpha , \neg \alpha \vdash $} implies
\mbox{$ \beta , \neg \alpha \vdash $} and, by {\bf Exchange} 
\mbox{$ \neg \alpha , \beta \vdash $},
i.e., \mbox{$ \neg \alpha \rightarrow \neg \beta $}
and similarly one shows \mbox{$ \neg \beta \rightarrow \neg \alpha $}.
Similarly, \mbox{$ \gamma , \alpha , \neg ( \gamma \wedge \alpha ) \vdash $} implies
\mbox{$ \gamma , \beta , \neg ( \gamma \wedge \alpha ) \vdash $} and 
\mbox{$ \gamma \wedge \beta \rightarrow \gamma \wedge \alpha \vdash $}.
The other cases are treated similarly.
\item
Assume \mbox{$ \beta \rightarrow \alpha $}.
For $\neg$: one has \mbox{$ \beta , \neg \alpha \vdash $}, 
\mbox{$ \neg \neg \beta , \neg \alpha \vdash $} and, by {\bf Exchange}, 
\mbox{$ \neg \alpha , \neg \neg \beta $}, i.e.,
\mbox{$ \neg \alpha \rightarrow \neg \beta $}.
For $\wedge$, consider the following derivation:

\( \begin{array}{rlrl}
{\bf \neg\vee 1} & \alpha , \gamma , \neg \gamma \vee \neg \alpha \vdash \\
\cline{2-2}
{\bf Circ} & \neg \gamma \vee \neg \alpha , \gamma , \alpha \vdash
& {\bf Assumption} & \neg \alpha , \beta \vdash \\
\cline{2-2} \cline{4-4}
{\bf Right-Weakening} & \neg \gamma \vee \neg \alpha , \gamma , \alpha , \beta \vdash 
& {\bf Left-Weakening} 
& \neg \gamma \vee \neg \alpha , \gamma , \neg \alpha , \beta  \vdash \\
\cline{2-3}
{\bf Cut} & \neg \gamma \vee \neg \alpha , \gamma , \beta \vdash \\
\cline{2-2}
{\bf Right-\wedge} & \neg ( \alpha \wedge \gamma ) , \beta \wedge \gamma \vdash
\end{array} \)

For $\vee$, the claim follows by duality from the two previous claims:
\[
\beta \rightarrow \alpha \Rightarrow \neg \alpha \rightarrow \neg \beta \Rightarrow 
\neg \alpha \wedge \neg \gamma \rightarrow \neg \beta \wedge \neg \gamma 
\Rightarrow 
\neg ( \neg \beta \wedge \neg \gamma ) \rightarrow \neg ( \neg \beta \wedge
\neg \gamma ) \Rightarrow 
\gamma \vee \beta \rightarrow \gamma \vee \alpha.
\]
\nopagebreak
\end{enumerate}
\end{proof}

The following lemma presents interesting properties of the implication and equivalence relations.
\begin{lemma} \label{the:equiv-ppties}
for any \mbox{$ \alpha , \beta \in {\cal L} $}:
\begin{enumerate}
\item \label{commute-ded}
Disjunction between orthogonal propositions is commutative: 
if \mbox{$ \alpha , \beta \vdash $}, then, 
\mbox{$ \alpha \vee \beta \simeq \beta \vee \alpha $},
\item \label{imp-and1}
\mbox{$ \beta \wedge \alpha \rightarrow \alpha $},
\item \label{imp-and2}
if \mbox{$ \beta \rightarrow \alpha $}, then, 
\mbox{$ \alpha \wedge \beta \simeq \beta \wedge \alpha \simeq \beta $}.
\end{enumerate}
\end{lemma}
Item~\ref{commute-ded} is a parallel to the last claim of Lemma~\ref{the:commute}.
\begin{proof}
\begin{enumerate}
\item
Derivation1:

\( \begin{array}{rlrl}
{\bf Logical Axiom} & \alpha , \neg \alpha \vdash & {\bf Assumption} & \alpha , \beta \vdash \\
\cline{2-4} 
{\bf Stuttering} & \alpha , \neg \beta , \neg \alpha \vdash \\
\end{array} \)

Derivation2:

\( \begin{array}{rlrl}
{\bf Logical Axiom} & \beta , \neg \beta \vdash \\
\cline{2-2}
{\bf Left-Right-Weakening} & \neg \alpha , \beta , \neg \beta , \neg \alpha \vdash 
& {\bf Derivation1} & \alpha , \neg \beta , \neg \alpha \vdash \\
\cline{2-4}
& {\bf \vee-Intro} & \alpha \vee \beta , \neg \beta , \neg \alpha \vdash \\
\cline{3-3}
& {\bf Right-\wedge} & \alpha \vee \beta , \neg \alpha \wedge \neg \beta \vdash \\
\cline{3-3}
& & \alpha \vee \beta , \neg ( \beta \vee \alpha ) \vdash
\end{array} \)
\item
Derivation:

\( \begin{array}{lr}
{\bf Logical Axiom} & \alpha , \neg \alpha \vdash \\
\cline{2-2}
{\bf Left-Weakening} & \beta , \alpha , \neg \alpha \vdash \\
\cline{2-2}
{\bf Left-\wedge} & \beta \wedge \alpha , \neg \alpha \vdash
\end{array} \)
\item
First, let us show that \mbox{$ \beta \wedge \alpha \rightarrow \beta $}:

\( \begin{array}{rlrl}
{\bf Assumption} & \beta , \neg \alpha \vdash 
& {\bf Logical Axiom} & \beta , \neg \beta \vdash \\
\cline{2-4}
& {\bf Stuttering} & \beta , \alpha , \neg \beta \vdash \\
\cline{3-3}
& {\bf Circ} & \neg \beta , \alpha , \beta \vdash \\
\cline{3-3}
& {\bf Right-\wedge} & \neg \beta , \beta \wedge \alpha \vdash \\
\cline{3-3}
& {\bf Exchange} & \beta \wedge \alpha , \neg \beta \vdash
\end{array} \)

Now, we show that \mbox{$ \beta \rightarrow \beta \wedge \alpha $}:

\( \begin{array}{rlrl}
{\bf Logical Axiom} & \beta , \neg \beta \vdash 
& {\bf Assumption} & \beta , \neg \alpha \vdash \\
\cline{2-4}
{\bf Stuttering} & \beta , \alpha , \neg \beta \vdash 
& {\bf Assumption} & \beta , \neg \alpha \vdash \\
\cline{2-4}
& {\bf \vee-Intro} & \beta , \neg \alpha \vee \neg \beta \vdash \\
\cline{3-3}
& & \beta , \neg ( \beta \wedge \alpha ) \vdash 
\end{array} \)
We have shown that \mbox{$ \beta \wedge \alpha \simeq \beta $}.
Let us show that \mbox{$ \alpha \wedge \beta \simeq \beta $}.
By item~\ref{imp-and1} above \mbox{$ \alpha \wedge \beta \rightarrow \beta $}.
We shall show now that \mbox{$ \beta \rightarrow \alpha \wedge \beta $}.

\( \begin{array}{rlrl}

& & {\bf Assumption} & \beta , \neg \alpha \vdash \\
\cline{4-4} 
{\bf Logical Axiom} & \beta , \neg \beta \vdash 
& {\bf Repetition} & \beta , \beta , \neg \alpha \vdash \\
\cline{2-4}
& {\bf \vee-Intro} & \beta , \neg \beta \vee \neg \alpha \vdash \\
\cline{3-3}
& & \beta , \neg ( \alpha \wedge \beta ) \vdash 
\end{array} \)
\end{enumerate}
\end{proof}

\subsection{Completeness} \label{sec:completeness}
We consider the deductive system $\cal R$ consisting of the ten rules described in 
Section~\ref{sec:soundness}.
From now on, \mbox{$ \Gamma \vdash $} will mean that the sequent can be derived 
by use of those ten rules only and \mbox{$ \Gamma \not \vdash $} 
will mean it cannot be obtained in this way.
The soundness claim of Theorem~\ref{the:soundness} guarantees that every derivable sequent
is valid in any O-space under any assignment.
Our goal is now to define an O-space and an assignment $v$ such that 
\mbox{$ v( \alpha ) = Z $} \emph{only if} the sequent 
\mbox{$ \alpha \vdash $} is derivable in the deductive system of the ten rules above.
Section~\ref{sec:logical-O-space} defines a structure on propositions that is an O-space.
Section~\ref{sec:assignment} defines an assignment of flats to propositions and 
presents the completeness result.

\subsubsection{The logical O-space} \label{sec:logical-O-space}
First, a notation.
\begin{definition} \label{def:D-def}
For any \mbox{$ \alpha \in {\cal L} $}, let 
\mbox{$ D_{\alpha} \eqdef \{ \beta \in {\cal L} \mid \beta , \neg \alpha \vdash \} $}. 
\end{definition}
We can now define the structure we are interested in.
\begin{definition} \label{def:top}
The structure  \mbox{$ \Omega = \langle {\cal L} , \bot , {\cal F} \rangle $} 
has carrier $\cal L$, the set of propositions, the relation $\bot$ is defined by 
\mbox{$ \alpha \, \bot \, \beta $} iff \mbox{$ \alpha , \beta \vdash $}, 
for any \mbox{$ \alpha , \beta \in {\cal L} $} and \mbox{$ \cal F = $} 
\mbox{$ \{ D_{\alpha} \mid \alpha \in {\cal L} \} $}.
\end{definition}

Our goal is now to show that $\Omega$ is an O-space.
Properties {\bf S} and {\bf Z} can be proved without ado.
\begin{lemma} \label{the:pre-space}
The structure $\Omega$ satisfies {\bf S} and {\bf Z} 
and \mbox{$ Z = \{ \alpha \in {\cal L} \mid \alpha \vdash \} $}.
\end{lemma}
\begin{proof}
\begin{enumerate}
\item \label{S-ded}
The relation $\bot$ is symmetric, since, by {\bf Exchange},
for any \mbox{$ \alpha , \beta \in {\cal L} $},
\mbox{$ \alpha , \beta \vdash $} iff \mbox{$ \beta , \alpha \vdash $}.
\item \label{Z-ded}
By definition, \mbox{$ Z =  \{ \alpha \in {\cal L} \mid \alpha , \alpha \vdash \} $}.
Since, by {\bf Contraction} and {\bf Left} (or {\bf Right}) {\bf-Weakening} 
\mbox{$ \alpha , \alpha \vdash $} iff \mbox{$ \alpha \vdash $}, we have 
\mbox{$ Z = \{ \alpha \in {\cal L} \mid \alpha \vdash \} $}.
By {\bf Right-Weakening}, if \mbox{$ \alpha \in Z $}, 
one has \mbox{$ \alpha , \beta \vdash $} for any \mbox{$ \beta \in {\cal L} $}, i.e.,
\mbox{$ \alpha \, \bot \, \beta $} for any $\beta$
and property {\bf Z} is satisfied.
\end{enumerate}
\end{proof}
Lemma~\ref{the:pre-space} shows that $\Omega$ satisfies the claims 
of Lemma~\ref{the:symmetric}. 
Some effort is needed in order to prove that $\Omega$ satisfies properties {\bf F}, 
{\bf O} and {\bf A}.
The next lemma characterizes the elements of $\cal F$, showing that the structure 
$\Omega$ satisfies {\bf F}, {\bf O} and {\bf A}.

\begin{lemma} \label{the:Omega-flats}
For any \mbox{$ \alpha \in {\cal L} $},
\begin{enumerate}
\item \label{single-flat}
\mbox{$ D_{\alpha} = $} \mbox{$ \overline{ \{ \alpha \} } $}, which shows that
any element of $\cal F$ is a flat and that, for any \mbox{$ \alpha \in {\cal L} $}, 
\mbox{$ \overline{\{ \alpha \}} \in {\cal F} $},
\item \label{Dneg}
\mbox{$ D_{\neg \alpha} = D_{\alpha}^\bot $}, which shows that $\cal F$ 
is closed under orthogonal complementation,
\item \label{DZ}
\mbox{$ D_{\neg ( \alpha \wedge \neg \alpha )} = Z $}, 
which shows that Z is a member of $\cal F$,
\item \label{Dotimes}
\mbox{$ D_{ \alpha \wedge \beta} = D_{\alpha} \otimes D_{\beta} $}, which shows
that $\cal F$ is closed under projection,
\item \label{Doplus}
\mbox{$ D_{\alpha \vee \beta} = D_{\alpha} \oplus D_{\beta} $},
\item \label{DO}
property {\bf O} is satisfied,
\item \label{DA}
property {\bf A} is satisfied.
\end{enumerate}
\end{lemma}
\pagebreak[2]
\begin{proof}
\begin{enumerate}
\item
Assume, first, that \mbox{$ \gamma \in \overline{\{ \alpha \}} $}.
For any \mbox{$ \beta \in \{ \alpha \}^\bot $}, we have 
\mbox{$ \gamma , \beta \vdash $}.
But \mbox{$ \neg \alpha \in \{ \alpha \}^\bot $} by {\bf Logical Axiom}.
Therefore \mbox{$ \gamma , \neg \alpha \vdash $}.
We must now show that, if \mbox{$ \beta , \neg \alpha \vdash $} 
we have \mbox{$ \beta \in \overline{\{\alpha\}} $}.
Assume \mbox{$ \gamma \in \{ \alpha \}^\bot $}.
Consider the following derivation.

\( \begin{array}{rlrl}
& & {\bf Assumption} & \beta , \neg \alpha \vdash \\
\cline{4-4}
{\bf Assumption} & \gamma , \alpha \vdash 
& {\bf Exchange} & \neg \alpha , \beta \vdash \\
\cline{2-2} \cline{4-4}
{\bf Right-Weakening} & \gamma , \alpha , \beta \vdash 
& {\bf Left-Weakening} & \gamma , \neg \alpha , \beta \vdash \\
\cline{2-4}
& {\bf Cut} & \gamma , \beta \vdash
\end{array} \)

We see that \mbox{$ \beta \, \bot \, \{ \alpha \}^\bot $},
i.e., \mbox{$ \beta \in \overline{ \{ \alpha \} } $}.
\item
Making use of Lemma~\ref{the:symmetric}, item~\ref{closure} and 
item~\ref{single-flat} above, we see that:
\[
D_{\neg \alpha} = \{ \beta \mid \beta , \neg \neg \alpha \vdash \} =
\{ \beta \mid \beta , \alpha \vdash \} = \{ \alpha \}^\bot = 
\overline{ \{ \alpha \} }^\bot = D_{\alpha}^\bot.
\]
\item
By {\bf Logical Axiom} we have \mbox{$ \alpha , \neg \alpha \vdash $} and,
by {\bf Left-$\wedge$}, \mbox{$ \alpha \wedge \neg \alpha \vdash $}.
By {\bf Left-Weakening}, \mbox{$ \beta , \alpha \wedge \neg \alpha \vdash $}
for any \mbox{$ \beta \in {\cal L} $}, i.e., 
\mbox{$ D_{\alpha \wedge \neg \alpha} = {\cal L} $}.
By Lemma~\ref{the:symmetric}, item~\ref{Xempty} and item~\ref{Dneg} above, 
we have \mbox{$ D_{\neg ( \alpha \wedge \neg \alpha )} = $}
\mbox{$ {\cal L}^\bot = Z $}.
\item
\mbox{$ D_{\alpha \wedge \beta} = $}
\mbox{$ \{ \gamma \mid \gamma , \neg ( \alpha \wedge \beta ) \vdash \}$}.
The following derivation shows that 
\mbox{$ D_{\alpha \wedge \beta} \subseteq D_{\beta} $}.

\( \begin{array}{rlrl}
& & {\bf Logical\  Axiom} & \beta , \neg \beta \vdash \\
\cline{4-4}
& & {\bf Left-Weakening} & \alpha , \beta , \neg \beta \vdash \\
\cline{4-4}
{\bf Assumption} & \gamma , \neg ( \alpha \wedge \beta ) \vdash 
& {\bf Right-\wedge} & \alpha \wedge \beta , \neg \beta \vdash \\
\cline{2-2} \cline{4-4}
{\bf Right-Weakening} & \gamma , \neg ( \alpha \wedge \beta ) , \neg \beta \vdash
& {\bf Left-Weakening} & \gamma , \alpha \wedge \beta , \neg beta \vdash \\
\cline{2-4}
& {\bf Cut} & \gamma , \neg \beta \vdash
\end{array} \)

Now, we want to prove that \mbox{$ D_{\alpha \wedge \beta} \subseteq $}
\mbox{$ {( D_{\beta} \cap D_{\neg \alpha} )}^\bot $}.
We have to show that from \mbox{$ \gamma , \neg ( \alpha \wedge \beta ) \vdash $},
\mbox{$ \delta , \neg \beta \vdash $} and \mbox{$ \delta , \alpha \vdash $}, one can
derive \mbox{$ \gamma , \delta \vdash $}.

Consider the following derivation:

\( \begin{array}{rlrl}
{\bf Assumption} & \delta , \neg \beta \vdash & {\bf Assumption} & \delta , \alpha \vdash \\
\cline{2-4}
{\bf Stuttering} & \delta , \beta , \alpha \vdash \\
\cline{2-2}
{\bf Right-\wedge} & \delta , \alpha \wedge \beta \vdash \\
\cline{2-2}
{\bf Exchange} & \alpha \wedge \beta , \delta \vdash 
& {\bf Assumption} & \gamma , \neg ( \alpha \wedge \beta ) \vdash \\
\cline{2-2} \cline{4-4}
{\bf Left-Weakening} & \gamma , \alpha \wedge \beta , \delta \vdash 
& {\bf Right-W...} & \gamma , \neg ( \alpha \wedge \beta ) , \delta \vdash \\
\cline{2-4}
& {\bf Cut} & \gamma , \delta \vdash
\end{array} \)

We have proved that \mbox{$ D_{ \alpha \wedge \beta } \subseteq $} 
\mbox{$ D_{\alpha} \otimes D_{\beta} $}.

Let us, now, assume that \mbox{$ \gamma \in D_{\alpha} \otimes D_{\beta} = $}
\mbox{$ D_{\beta} \cap {( D_{\beta} \cap D_{\neg \alpha} )}^\bot $}.
We want to prove that \mbox{$ \gamma \in D_{\alpha \wedge \beta} $}.
We see that \mbox{$ \gamma , \neg \beta \vdash $} and  
\mbox{$ \gamma , \delta \vdash $} for any $\delta$ such that 
\mbox{$ \delta , \neg \beta \vdash $} and \mbox{$ \delta , \alpha \vdash $}.
Consider, now, any \mbox{$ \epsilon \in \{ \alpha \wedge \beta \}^\bot $}.
We have \mbox{$ \epsilon , \alpha \wedge \beta \vdash $} and, 
by {\bf$\wedge$Right-Elim} \mbox{$ \epsilon , \beta , \alpha \vdash $}.
By {\bf$\wedge$-Left}, we have \mbox{$ \epsilon \wedge \beta , \alpha \vdash $}. 
The following derivation shows that 
\mbox{$ \epsilon \wedge \beta , \neg \beta \vdash $}.

\( \begin{array}{rl}
{\bf Logical Axiom} & \beta , \neg \beta \vdash \\
\cline{2-2}
{\bf Exchange} & \neg \beta , \beta \vdash \\
\cline{2-2} 
{\bf Right-Weakening} & \neg \beta , \beta , \epsilon \vdash \\
\cline{2-2}
{\bf Right-\wedge} & \neg \beta , \epsilon \wedge \beta \vdash 
\end{array} \)

We see that $\epsilon \wedge \beta$ is a suitable $\delta$, i.e., 
\mbox{$ \epsilon \wedge \beta , \neg \beta \vdash $} and 
\mbox{$ \epsilon \wedge \beta , \neg \beta \vdash $}.
Therefore \mbox{$ \gamma , \epsilon \wedge \beta \vdash $}.
By {\bf $\wedge$Right-Elim}, we have \mbox{$ \gamma , \beta , \epsilon \vdash $}.
But \mbox{$ \gamma , \neg \beta \vdash $} and one easily derives
\mbox{$ \gamma , \epsilon \vdash $}.
We have shown that \mbox{$ \gamma \in \{ \epsilon \}^\bot $} for any 
\mbox{$ \epsilon \in \{ \alpha \wedge \beta \}^\bot $} and therefore
\mbox{$ \gamma \in \overline{ \{ \alpha \wedge \beta \} } $}.
We conclude by item~\ref{single-flat} above.
\item
We have:
\[
D_{\alpha \vee \beta} = \{ \gamma \in {\cal L} \mid \gamma , 
\neg ( \alpha  \vee \beta ) \vdash \} = 
\{ \gamma \in {\cal L} \mid \gamma , \neg \beta \wedge \neg \alpha \vdash \} = 
\]
\[
D_{\neg ( \neg \beta \wedge \neg \alpha )} = 
D_{\neg \beta \wedge \neg \alpha}^{\bot} = 
{( D_{\beta}^\bot \otimes D_{\alpha}^\bot )}^\bot = 
D_{\alpha} \oplus D_{\beta}.
\]
\item
Since \mbox{$ A \in {\cal F} $}, \mbox{$ A = D_{\alpha} $} for some
\mbox{$ \alpha \in {\cal L} $}.
For any \mbox{$ \beta \in {\cal L} $}, 
\mbox{$ D_{\beta} \otimes D_{\alpha} \, \bot \, D_{\beta} \otimes D_{\neg \alpha} $} 
and, by Lemma~\ref{the:symmetric}, item~\ref{bot-union}, 
\mbox{$ ( D_{\beta} \otimes D_{\alpha} ) \oplus 
( D_{\beta} \otimes D_{\neg \alpha} ) = $}
\mbox{$ \overline{ ( D_{\beta} \otimes D_{\alpha} ) \cup 
( D_{\beta} \otimes D_{\neg \alpha} ) } $}.
By items~\ref{Dotimes} and~\ref{Doplus} above and, then, by item~\ref{single-flat} 
and Lemma~\ref{the:symmetric}, item~\ref{bar-cup}, 
we have \mbox{$ D_{( \beta \wedge \alpha ) \vee ( \beta \wedge \neg \alpha )} = $}
\mbox{$ \overline{D_{\beta \wedge \alpha} \cup D_{\beta \wedge \neg \alpha}} = $}
\mbox{$ \overline{ \{ \beta \} \otimes A \cup 
\{ \beta \} \otimes A^\bot } $}.
We shall now show that
\mbox{$ \beta \in D_{ ( \beta \wedge \alpha ) \vee ( \beta \wedge \neg \alpha ) } $} 
i.e., \mbox{$ \beta , \neg ( ( \beta \wedge \alpha ) \vee ( \beta \wedge \neg \alpha ) ) \vdash $}.
First, consider the derivation:

\( \begin{array}{rlrl}
{\bf Logical Axiom} & \alpha ,  \neg \alpha \vdash \\
\cline{2-2}
{\bf Exchange} & \neg \alpha , \alpha \vdash \\
\cline{2-2}
{\bf Right-Weakening} & \neg \alpha , \alpha , \beta \vdash \\
\cline{2-2}
{\bf Right-\wedge} & \neg \alpha , \beta \wedge \alpha \vdash \\ 
\cline{2-2}
{\bf Left-Weakening} & \beta , \neg \alpha , \beta \wedge \alpha \vdash \\
\cline{2-2}
& \beta , \neg \alpha , \neg ( \neg \alpha \vee \neg \beta ) \vdash 
& {\bf \neg\vee 1} & \beta , \neg \alpha , \alpha \vee \neg \beta \vdash \\
\cline{2-4}
{\bf Stuttering} 
& \beta , \neg \alpha , \neg \alpha \vee \neg \beta , \alpha \vee \neg \beta \vdash
\end{array} \)

Then consider:

\( \begin{array}{rl}
{\bf \neg\vee 1} & \beta , \alpha , \neg \alpha \vee \neg \beta \vdash \\
\cline{2-2}
{\bf Right-Weakening} & 
\beta , \alpha , \neg \alpha \vee \neg \beta , \alpha \vee \neg \beta \vdash 
\end{array} \)

By {\bf Cut} we obtain
\mbox{$ \beta , \neg \alpha \vee \neg \beta , \alpha \vee \neg \beta \vdash $}.
By {\bf Right-$\wedge$}, then
\[
\beta , ( \alpha \vee \neg \beta ) \wedge ( \neg \alpha \vee \neg \beta ) \vdash ,
{\rm \ i.e., \ }
\beta , \neg ( ( \beta \wedge \alpha ) \vee ( \beta \wedge \neg \alpha ) ) \vdash.
\]
We have shown that \mbox{$ \beta \in D_{( \beta \wedge \alpha) \vee 
(\beta \wedge \neg \alpha ) } $}.

For the second claim, we need to show that 
\mbox{$ D_{\beta \wedge \alpha} \subseteq $} 
\mbox{$ \overline{ D_{\beta}  \cup D_{\beta \wedge \neg \alpha } } $}, 
equivalently, to show the validity of the rule

\( \begin{array}{ccc}
\gamma , \neg ( \beta \wedge \alpha ) \vdash & \delta , \neg \beta \vdash 
& \delta , \neg ( \beta \wedge \neg \alpha ) \vdash \\
\hline
& \gamma , \delta \vdash
\end{array} \)

Derivation1:

\( \begin{array}{rl}
{\bf Assumption} & \delta , \alpha \vee \neg \beta \vdash \\
\cline{2-2}
{\bf \vee Right-Elim} & \delta , \alpha \vdash \\
\cline{2-2}
{\bf Right-Weakening} & \delta , \alpha , \gamma \vdash 
\end{array} \)

Derivation2:

\( \begin{array}{rlrl}
{\bf Assumption} & \gamma , \neg \alpha \vee \neg \beta \vdash \\
\cline{2-2} 
{\bf \vee Right-Elim} & \gamma , \neg \alpha \vdash \\
\cline{2-2}
{\bf Exchange} & \neg \alpha , \gamma 
& {\bf Assumption} & \delta , \neg \beta \vdash\\
\cline{2-2} \cline{4-4}
{\bf Left-W.} & \delta , \beta , \neg \alpha , \gamma \vdash 
& {\bf Right-W.} & \delta , \neg \beta , \neg \alpha , \gamma \vdash \\
\cline{2-4}
{\bf Cut} & \delta , \neg \alpha , \gamma \vdash 
& {\bf Derivation1} & \delta , \alpha , \gamma \vdash \\
\cline{2-4}
& {\bf Cut} & \delta , \gamma \vdash
\end{array} \)
\item
Since \mbox{$ A \in {\cal F} $}, \mbox{$ A = D_{\alpha}$} for some 
\mbox{$ \alpha \in {\cal L} $} and, by item~\ref{single-flat} above,
\mbox{$ A = \overline{\{ \alpha \}} $}.
By Lemma~\ref{the:symmetric}, \mbox{$ A \otimes B = $} 
\mbox{$ \{ \alpha \} \otimes B $}, proving our claim.
\end{enumerate}
\end{proof}

We may, now, conclude.
\begin{theorem} \label{the:logical-O-space}
The structure $\Omega$ is an O-space.
\end{theorem}

\subsubsection{An assignment of flats to propositions} \label{sec:assignment}
In the O-space $\Omega$ we define an assignment $v$ that assigns an element of
$\cal F$ to any atomic proposition. 
For any \emph{atomic} proposition $\sigma$, let
\begin{equation} \label{eq:vdef}
v(\sigma) \eqdef D_{\sigma}.
\end{equation}
Note that $v(\sigma)$ is indeed an element of $\cal F$.
As explained in Definition~\ref{def:interpretation} the assignment function $v$ 
generalizes to arbitrary elements of $\cal L$. 
Note that the assignment function $v$ assigns a set of propositions to any proposition.

\begin{lemma} \label{the:v}
For any proposition \mbox{$\alpha \in {\cal L} $}, \mbox{$ v ( \alpha ) = $}
\mbox{$ D_{\alpha} $}.
\end{lemma}
\begin{proof}
Note that the claim implies that \mbox{$v ( \alpha ) $} is an element of $\cal F$, 
by Lemma~\ref{the:Omega-flats}.
We reason by induction on the structure of $\alpha$. 
\begin{itemize}
\item
For an atomic proposition $\sigma$, the claim about \mbox{$v ( \sigma ) $} 
holds by construction, see Equation~(\ref{eq:vdef}).
\item
For the negation of an atomic proposition, \mbox{$ v ( \neg \sigma ) =$} 
\mbox{$ { v ( \sigma ) }^\bot $} by definition and we conclude by 
Lemma~\ref{the:Omega-flats}, item~\ref{Dneg}.
\item
For a proposition \mbox{$ \alpha \wedge \beta $}, by definition and the induction hypothesis,
\mbox{$ v ( \alpha \wedge \beta ) = $} 
\mbox{$ v( \alpha ) \otimes v ( \beta ) = $} \mbox{$ D_{\alpha} \otimes D_{\beta} $}
and we conclude by Lemma~\ref{the:Omega-flats}, item~\ref{Dotimes}.
\item
For a proposition \mbox{$ \alpha \vee \beta $}, similarly, 
by Lemma~\ref{the:Omega-flats}, item~\ref{Doplus}.
\end{itemize}
\end{proof}

\begin{theorem} [Completeness] \label{the:omega}
Any sequent valid in the O-space $\Omega$ under the assignment $v$ can be derived in
the deductive system $\cal R$ defined in Section~\ref{sec:soundness}.
\end{theorem}
\begin{proof}
Assume \mbox{$ \alpha_{0} , \ldots , \alpha_{n - 1} \models $} is valid in $\Omega$ under $v$.
We have 
\[
v ( ( \ldots ( \alpha_{0} \wedge \alpha_{1} ) \wedge \ldots ) \wedge \alpha_{n - 1} ) = Z
\]
and, by Lemma~\ref{the:v},
\mbox{$ D_{ ( \ldots ( \alpha_{0} \wedge \alpha_{1} ) \wedge \ldots ) \wedge \alpha_{n - 1} } = Z $}.
By {\bf Logical Axiom}, now, for any \mbox{$ \alpha \in {\cal L} $},
\mbox{$ \alpha \in D_{\alpha} $} and therefore we have
\mbox{$ ( \ldots ( \alpha_{0} \wedge \alpha_{1} ) \wedge \ldots ) \wedge 
\alpha_{n - 1} \in Z $}. 
We see that
\mbox{$ ( \ldots ( \alpha_{0} \wedge \alpha_{1} ) \wedge \ldots ) \wedge \alpha_{n - 1} \vdash $} and, by {\bf $\wedge$Left-Elim},
\mbox{$ \alpha_{0} , \ldots , \alpha_{n - 1} \vdash $}.
\end{proof}
We may now summarize.
\begin{theorem} \label{the:completeness}
The deductive system $\cal R$ described in Section~\ref{sec:soundness} is sound and
complete for O-spaces:
a sequent is valid ( in any O-space, under any assignment ) iff it can be derived 
in the system $\cal R$.
\end{theorem}
\begin{proof}
By Theorems~\ref{the:soundness} and~\ref{the:omega}.
\end{proof}

\section{Future research} \label{sec:future}
Here are a number of topics that deserve further attention.
\begin{itemize}
\item
Is system $\cal R$ complete for Hilbert spaces and projection, i.e., are there
rules of inference that are valid in Hilbert spaces (and therefore in Quantum Logic) 
but cannot be derived in $\cal R$?
\item
Can one prove a Cut elimination theorem and/or some normalization theorem for 
the proof system of Section~\ref{sec:rules}?
\item
Morphisms in O-spaces: what are the families of morphisms that generalize
linear and auto-adjoint functions?
\item
Mixed states: can a mixed state be modeled by an element of an O-space?
\item
Since the structures presented above include classical, quantic and all sorts of 
intermediate structures, they seem to be ideal for studying the boundaries between 
classical and quantic behavior.
\item
The $\otimes$ connective is some kind of \emph{and then} connective. 
Can this non-commutative conjunction help model conjunction in natural language? 
\item
In an effort to understand the concepts of Quantum Physics, the description of its 
logic underpinnings can only be a small start.
It can only be expected to provide to QP what is provided to classical physics
by the logic of elementary set theory, or boolean algebras, not much more than
a basic language.
Structures richer than O-spaces are necessary for doing Quantum Physics.
Such structures must provide a measure of the probability of the different orthogonal
possible results of a measurement.
They must also be capable of expressing the dynamics of a quantic system that is
described by the Schr\"{o}dinger equation.
Very preliminary ideas can be found in Section 7 of~\cite{Lehmann-metalinear1:2022}.
\end{itemize}

\section{Acknowledgments} \label{sec:ack}
My deepest thanks to Kurt Engesser who directed me to the volume of 
John von Neumann's letters~\cite{vNeumann_letters} edited by Mikl\'{o}s R\'{e}dei.
Von Neumann's letter dated Nov. 13, Wednesday, 1935, addressed to Garret
Birkhoff during the elaboration of~\cite{BirkvonNeu:36} convinced me that my intuitions
were worth pursuing.
The letter is quoted in~\cite{Lehmann-metalinear1:2022}.
Thanks also to Jean-Marc Levy-Leblond who straightened me out on QM and Michael Freund, 
Daniel Weil and Menachem Magidor who commented on drafts of this paper.
\bibliographystyle{plain}

\end{document}